\documentclass{article} 
\usepackage{jmlr2e}

\usepackage[colorlinks,urlcolor=cyan,citecolor=blue,linkcolor=magenta]{hyperref} 
\usepackage{hyperref}
\usepackage{url}
\usepackage{physics}

\usepackage{booktabs}
\usepackage{multicol}
\usepackage[export]{adjustbox}
\usepackage{multirow}
\usepackage{tabularx}
\usepackage{makecell}
\usepackage[shortlabels,inline]{enumitem}

\usepackage{amsmath,amsfonts,bm,amsthm}
\usepackage{thmtools}
\usepackage{thm-restate}

\def\cO{{\mathcal{O}}}
\def\cP{{\mathcal{P}}}

\def\bR{{\mathbb{R}}}

\theoremstyle{plain} \numberwithin{equation}{section}
\newtheorem{theorem}{Theorem}[section]
\numberwithin{theorem}{section}
\newtheorem{lemma}[theorem]{Lemma}

\newtheorem{proposition}[theorem]{Proposition}

\newtheorem{conjecture}[theorem]{Conjecture}

\theoremstyle{definition}
\newtheorem{definition}[theorem]{Definition}

\newtheorem{remark}[theorem]{Remark}
\newtheorem{example}[theorem]{Example}

\newtheorem{algorithm}[theorem]{Algorithm}
\theoremstyle{plain}

\newcommand{\ring}{\mathcal{P}}

\newcommand{\rringle}[1]{\mathcal{P}_{\le {#1}}}
\newcommand{\fullring}{\mathbb{R}[x_1,\ldots, x_n]}
\newcommand{\ideal}[1]{\langle{#1}\rangle}

\newcommand{\xs}{x_1,\ldots, x_n}

\newcommand{\Span}[1]{\mathrm{span}({#1})}
\newcommand{\videal}[1]{\mathcal{I}({#1})}
\newcommand{\degree}[1]{\mathrm{deg}({#1})}
\newcommand{\points}{\bm{x}_1, \ldots, \bm{x}_m}
\newcommand{\gnvec}[2]{\mathfrak{n}_{\mathrm{g}, {#1}}({#2})}

\usepackage{xcolor} %
\usepackage{tikz}
\usetikzlibrary{fit,calc}

\colorlet{pink}{red!40}
\colorlet{lightblue}{blue!30}
\colorlet{lightgreen}{green!30}

\title{Monomial-agnostic computation of vanishing ideals}

\author{
       \name \hspace{-3mm} Hiroshi Kera \email \texttt{\href{mailto:kera@chiba-u.jp}{kera@chiba-u.jp}} \\
       \addr 
       Graduate School of Engineering, \\
        Chiba University, \\
        11-33 Yayoi-cho, Inage-ku, Chiba-Shi, Chiba, Japan
       \AND
       \name Yoshihiko Hasegawa \email \texttt{\href{mailto:hasegawa@biom.t.u-tokyo.ac.jp}{hasegawa@biom.t.u-tokyo.ac.jp}} \\
       \addr 
        Graduate School of Information Science and Technology, \\
        The University of Tokyo, \\
        7-3-1 Hongo, Bunkyo-ku, Tokyo, Japan
       }

\begin{document}

\maketitle

\begin{abstract}
In the last decade, the approximate basis computation of vanishing ideals has been studied extensively in computational algebra and data-driven applications such as machine learning. However, symbolic computation and the dependency on term order remain essential gaps between the two fields. In this study, we present the first \textit{monomial-agnostic} basis computation, which works fully numerically with proper normalization and without term order. 
This is realized by gradient normalization, a newly proposed data-dependent normalization that normalizes a polynomial with the magnitude of gradients at given points. The data-dependent nature of gradient normalization brings various significant advantages: i) efficient resolution of the spurious vanishing problem, the scale-variance issue of approximately vanishing polynomials, without accessing coefficients of terms, ii) scaling-consistent basis computation, ensuring that input scaling does not lead to an essential change in the output, and iii) robustness against input perturbations, where the upper bound of error is determined only by the magnitude of the perturbations. Existing studies did not achieve any of these. As further applications of gradient information, we propose a monomial-agnostic basis reduction method and a regularization method to manage positive-dimensional ideals. 
\end{abstract}

\section{Introduction}\label{sec:introduction}
Let $X\subset\bR^n$ be a finite set of points. The vanishing ideal $\videal{X}\subset \fullring$ of $X$ is the set of polynomials that vanish for $X$:
\begin{align}
    \videal{X} = \qty{ g\in\fullring \mid \forall \bm{x} \in X, g(\bm{x}) = 0}.
\end{align}
There are various algorithms to compute generators of a vanishing ideal. Particularly, approximate computation of generators, where approximate vanishing $\abs{g(\bm{x})} \le \epsilon$ for a potentially perturbed point $\bm{x} \in X$ and predesignated threshold $\epsilon \ge 0$ are considered, have been recently studied in computational algebra~\citep{abbott2008stable,heldt2009approximate,fassino2010almost,robbiano2010approximate,fassino2013simple,limbeck2013computation,kera2022border} and further introduced to various data-centric applications such as machine learning~\citep{livni2013vanishing,kiraly2014dual,hou2016discriminative,kera2016vanishing,kera2018approximate,wirth2022conditional,wirth2023approximate}, computer vision~\citep{zhao2014hand,yan2018deep}, robotics~\citep{iraji2017principal,antonova2020analytic}, nonlinear systems~\citep{torrente2009application,kera2016noise,karimov2020algebraic,karimov2023identifying}, and signal processing~\citep{wang2018nonlinear, wang2019polynomial}. 

The approximate setup and the new demands from such applications raise interesting problems that do not have to be considered in the exact setup in computational algebra. For example, approximate vanishing $\abs{g(\bm{x})} \le \epsilon$ raises the fundamental need for normalization of polynomial because $\alpha\cdot g$ can be arbitrarily turned to be either approximately vanishing or not by changing the scaling $\alpha \in \mathbb{R}$. Such a scale-variant nature does not occur in the exact case with $\epsilon = 0$. Another example is that the dependence of algorithms and generators on a term order is considered unfavorable in many of the applications mentioned above. 
In such applications, indeterminates or variables have specific meanings (e.g., temperature and height), and the term order inevitably injects an uneven variable importance, biasing the results of data analysis. Besides, many applications are ruled by numerical computation, which allows a considerable hardware acceleration using GPUs rather than symbolic computation; thus, heavy reliance on computer-algebraic tools loses the consistency of the implementation. 
Consequently, vanishing component analysis~(VCA;~\cite{livni2013vanishing}), a term-order-free, numerical algorithm, is more widely used than those developed in computational algebra. 

While the practical utility of VCA and its variants has been reported in various studies, \cite{kera2019spurious} pointed out that these algorithms do not equip normalization and thus theoretically suffer from the aforementioned scale-variant issue. These algorithms construct a polynomial $g \in \fullring$ from an $\fullring$-combination of several polynomials,
\begin{align}
    g =p_1 h_1 + \cdots p_s h_s, \quad \text{where}\ \  p_i, h_i \in \fullring,\ \ i= 1, \ldots, s.
\end{align}
Polynomials $h_1, \ldots, h_s$ are also constructed similarly. Polynomial expansions and arrangement of terms are necessary if one wants to perform a coefficient normalization, a widely used method that normalizes a polynomial to have the unit square sum of term coefficients. As these cannot be straightforwardly done by numerical computation, VCA and its variants do not equip coefficient normalization. It is also worth noting that even if possible, it is still not trivial how one can integrate the normalization with the optimization problems solved in these algorithms for constructing the best approximately vanishing polynomials. \cite{kera2019spurious} proposed a framework to incorporate coefficient normalization into VCA-family algorithms; however, it theoretically suffers from the exponential time and space complexity in the number of variables because of the exponential growth of the length of coefficient vectors.

In this study, we achieve the first efficient, numerical, and term-order-free approximate basis computation---or \textit{monomial-agnostic} basis computation---by introducing gradient normalization, a novel data-driven normalization. Specifically, gradient normalization of polynomial $g \in \fullring$ over a finite point set $X \subset \mathbb{R}^n$ is defined by 
\begin{align}
    \frac{g}{\sqrt{\sum_{\bm{x} \in X} \norm{\nabla g(\bm{x})}^2_2}}, \quad \text{where}\ \ \nabla g(\bm{x}) = \qty( \frac{\partial g}{\partial x_1} (\bm{x}), \ldots, \frac{\partial g}{\partial x_n}(\bm{x}))^{\top} \in \mathbb{R}^n.
\end{align}
The norm $\norm{\ \cdot\ }_2$ denotes the Euclidean norm. The denominator $\sqrt{\sum_{\bm{x} \in X} \norm{\nabla g(\bm{x})}^2_2}$, or the \textit{gradient semi-norm}, may vanish for $g \ne 0$, and thus the normalization is not always valid. However, we can show that the normalization is valid for all the non-constant polynomials that we need to handle in VCA-family algorithms~\citep{livni2013vanishing,kiraly2014dual,hou2016discriminative,kera2018approximate,yan2018deep}. Gradient normalization can be computed efficiently and fully numerically by reusing the computation in the algorithms. Further, its data-driven nature leads to more robust computation against point perturbation or scaling, which cannot be attained by coefficient normalization. 
Notably, with gradient normalization, the output of VCA for $X_{\alpha} = \{\alpha\cdot \bm{x}_1, \ldots, \alpha\cdot\bm{x}_m\}$ and $\alpha\cdot \epsilon$ returns consistent for all $\alpha > 0$; the outputs for distinct $\alpha_1, \alpha_2 \ne 0$ are generally different (reflecting the scale difference) but there is an explicit rule to transform one another. This is not the case with coefficient normalization but is crucial for robust computation with real data points as the scaling of points is a common preprocessing for computer-algebraic basis computation algorithms and data-centric applications so that one can avoid numerical errors caused by the underflow or overflow. 
Including such a scale consistency, this study theoretically and empirically investigates the advantages of gradient normalization over coefficient normalization in the context of approximate computation of generators of vanishing ideal.\footnote{The present study is the extended version of a conference paper~\citep{kera2020gradient} with rich additional contents. Examples are as follows: 
(i) complete proofs of the relevant theorems, propositions, and lemmas, 
(ii) new analysis that shows normalization leads to a reduced basis size~(Proposition~\ref{prop:ub-basis-size}),
(iii) the scaling \textit{inconsistency} of the coefficient normalization~(Proposition~\ref{prop:scaling-inconsistency}),
 (iv) a perturbation analysis~(Proposition~\ref{prop:perturbation}), 
(v) a preprocessing method for faster gradient normalization~(Appendix~\ref{app:preprocessing}), 
(vi) a new heuristics to terminate the basis computation for positive-dimensional ideals~(Section~\ref{sec:positive-ideal}), 
(vii) extensive numerical experiments to validate the advantages of gradient normalization~(Section~\ref{sec:experiments}).
(viii) an open-source Python library MAVI for the fully numerical monomial-agnostic computation of vanishing ideals (\url{https://github.com/HiroshiKERA/monomial-agnostic-vanishing-ideal}), which supports three popular backends, NumPy~\citep{numpy}, JAX~\citep{jax}, and PyTorch~\citep{pytorch}, and runs with GPUs. 
}

\paragraph{Organization of the paper.} Section~\ref{sec:preliminaries} sets up basic definitions, notations, and background knowledge associated with approximate vanishing polynomials, such as the normalized VCA and the spurious vanishing problem. Section~\ref{sec:gradient-normalization} introduces gradient normalization and investigates its validity and computation. Then, Section~\ref{sec:advantages-of-gradient-normalization} presents several advantageous properties of gradient normalization. Gradient information is beneficial not only for normalization, which is discussed in Section~\ref{sec:futher-applications}. Lastly, Section~\ref{sec:experiments} demonstrates the effectiveness of gradient normalization through numerical experiments. 

\section{Preliminaries}\label{sec:preliminaries}
\subsection{Basic definitions and notations}
Throughout the paper, we focus on the polynomial ring $\ring = \mathbb{R}[x_1,x_2,\ldots,x_n]$ with indeterminates $\xs$ and a finite set of points $X \subset \bR^n$. We denote the restrictions of a polynomial set $H \subset \ring$ to degree $t$ and degree up to $t$ by $H_{t}$ and $H_{\le t}$, respectively. A set of several degree-$t$ polynomials may be also denoted with a subscript $t$ (e.g., $C_t$).
\begin{definition}\label{def:vanishing-ideal}
The \textbf{vanishing ideal} $\videal{X}\subset\ring$ of a finite set $X\subset\mathbb{R}^n$ is the set of all polynomials that take the zero value (i.e., vanish) for any point in $X$:
\begin{align}
\videal{X} & =\left\{ g\in\ring\mid\forall\bm{x}\in X,g(\bm{x})=0\right\}.
\end{align}
\end{definition}
\begin{definition}
Given a set of points $X = \{\points\} \subset \mathbb{R}^n$, the \textbf{evaluation vector} of a polynomial $h\in\ring$ with respect to $X$ is defined by
\begin{align}
h(X) & =\begin{pmatrix}h(\bm{x}_{1}) & h(\bm{x}_{2}) & \cdots & h(\bm{x}_m)\end{pmatrix}^{\top}\in\mathbb{R}^{m}.
\end{align}
For a set of polynomials $H=\left\{ h_{1},h_{2},\ldots,h_r\right\} \subset\ring$,
the \textbf{evaluation matrix} of $H$ with respect to $X$ is defined by
\begin{align}
	H(X) & = \left(\begin{array}{cccc}h_{1}(X) & h_{2}(X) & \cdots & h_r(X)\end{array}\right)\in\mathbb{R}^{m\times r}. 
\end{align}
The gradient of $h$ at $\bm{x}$ will be denoted by $\nabla h(\bm{x}) = \qty((\partial{h}/{\partial x_1})(\bm{x}),\ldots, (\partial{h}/{\partial x_n})(\bm{x}) )^{\top} \in \bR^n$.
\end{definition}
As shown by the definition of the vanishing ideal, we are interested in the evaluation values of polynomials at the given set of points $X$. Hence, a polynomial $h$ can be represented by its evaluation vector $h(X)$, which link the product and weighted sum of polynomials with linear-algebraic operations. The product of $h_1,h_2\in \ring$ corresponds to $h_1(X)\odot h_2(X)$, where $\odot$ denotes the element-wise product; the weighted sum $w_{1}h_1 + w_2h_2$ with $w_1,w_2\in\mathbb{R}$ corresponds to $w_{1}h_1(X) +  w_{2}h_2(X)$. For convenience, we define the following notations. 
\begin{definition}
 Let  $H=\{h_1,h_2,\ldots,h_r\}\subset \ring$ be a set of $r$ polynomials, and let $W=\mqty(\bm{w}_1 & \bm{w}_2 & \cdots \bm{w}_{s})\in\mathbb{R}^{r\times s}$ with $\bm{w}_i = (w_{i1}, \ldots, w_{ir})^{\top} \in \bR^r$ for $i=1,\ldots, s$.
 The products $H\bm{w}_i$ and $HW$ are respectively defined by 
 \begin{align}
 	H\bm{w}_i =\sum_{j=1}^{r}w_{ij} h_j 
  \quad\text{ and }\quad
 	HW =\{H\bm{w}_1,H\bm{w}_2,\ldots,H\bm{w}_{s}\}.
 \end{align}
 Note that we have $(H\bm{w})(X) = H(X)\bm{w}$ and $(HW)(X) = H(X)W$ using these notations. 
\end{definition}
\begin{definition} Let $X \subset \bR^n$ be a finite set of points.
A polynomial $g\in\ring$ is \textbf{$\epsilon$-approximately vanishing} for $X$ if $\norm{g(X)}_2\le\epsilon$, where $\epsilon \ge 0$, and $\norm{\,\cdot\,}_2$ denotes the Euclidean norm. Otherwise, $g$ is called \textbf{$\epsilon$-nonvanishing}.
\end{definition}
We may drop ``$\epsilon$-'' when no specific $\epsilon$ is in mind. When $\epsilon = 0$, we may use \textit{exact vanishing} for emphasis.

\paragraph{Other definitions and notations}
Let $H = \{h_1, \ldots, h_r\} \subset \ring$ and $X = \{\points\} \subset \bR^n$. The $\bR$-span of $H$ is denoted by $\Span{H} = \{\sum_{i=1}^r v_i h_i\mid v_1, \ldots, v_r\in\mathbb{R}\}$. 
We also use the same notation to denote the column space of matrix, e.g., $\Span{H(X)} = \{\sum_{i=1}^s v_ih(X) \mid v_1, \ldots, v_r\in\mathbb{R}\}$.
The ideal generated by $H$ is denoted by $\ideal{H} = \qty{ \sum_{i=1}^r q_i h_i \mid q_1, \ldots, q_r \in \ring }$.
The total degree of polynomial $h\in\ring$ is denoted by $\mathrm{deg}(h)$, and the degree with respect to $x_k$ is denoted $\mathrm{deg}_k(h)$. The cardinality of set $A$ will be denoted by $\abs{A}$. We denote by $\mathrm{diag}(d_1, \ldots, d_s)$ the diagonal matrix of size $s$ with $d_1, \ldots, d_s$ along its diagonals. The identity matrix of size $s$ is denoted by $E_s \in \bR^{s\times s}$. 

\subsection{The normalized VCA}
Our idea of using gradients is general enough to be integrated with several basis computation algorithms of vanishing ideals. However, to avoid an unnecessarily abstract discussion, we focus on the normalized version of VCA~\citep{kera2019spurious}.\footnote{While this algorithm was originally named the simple basis construction algorithm by ourselves, we renamed it as the normalized VCA because this gives a better intuition on the algorithm.} 
The normalized VCA is free of term order and the spurious vanishing problem. Given a finite point set $X\subset\bR^n$ and a threshold $\epsilon \ge 0$, the algorithm proceeds from lower to higher degree, constructing two degree-$t$ polynomial sets $F_t, G_t \subset \ring$ at the $t$-th iteration based on the results at the lower degree. The former set $F_t$ collects nonvanishing polynomials, while the latter $G_t$ collects vanishing polynomials.

\begin{algorithm}{\bf (The normalized VCA)}\label{alg:NVCA}\\
Let $X\subset\mathbb{R}^n$ be a finite set of points and $\epsilon\ge 0 $ be a fixed threshold. Let $ F = \{f_0\}, G = \{\,\} \subset\ring$, where $f_0$ may be any nonzero constant polynomial. Perform the following procedures for $t = 1,2,\ldots$ until the termination criterion is met. 
\begin{enumerate}
\item[{\bf S1}] 
Set up pre-candidate polynomials of degree $t$ by $C_1^{\mathrm{pre}}=\{x_1,x_2,\ldots,x_n\}$ if $t=1$ and otherwise,
\begin{align}
    C_t^{\mathrm{pre}} = \{pq \mid p\in F_1, q\in F_{t-1}\}.
\end{align}
 Then, construct the final candidate polynomials $C_t$ through the following projection,
\begin{align}\label{eq:orthogonalization}
     C_{t} &= C_{t}^{\mathrm{pre}} - F_{\le t-1}F_{\le t-1}(X)^{\dagger}C_{t}^{\mathrm{pre}}(X),
 \end{align}
 where $\,\cdot\,^{\dagger}$ denotes the pseudo-inverse of matrix.
 \item[{\bf S2}] Solve the following generalized eigenvalue problem: 
\begin{align}\label{eq:gep}
    C_t(X)^{\top}C_t(X)V = \mathfrak{n}(C_t)^{\top}\mathfrak{n}(C_t)V\Lambda, 
\end{align}
with generalized eigenvectors $V = \mqty(\bm{v}_1, \ldots, \bm{v}_{\abs{C_t}})$, generalized eigenvalues $\Lambda = \mathrm{diag}(\lambda_1,\lambda_2,\ldots,\lambda_{|C_t|})$, and a normalization matrix $\mathfrak{n}(C_t)\in\mathbb{R}^{\ell\times |C_t|}$, which will be introduced shortly~(cf.~Definition~\ref{def:normalization-mapping}).
\item[{\bf S3}] Construct degree-$t$ basis polynomials by linearly combining polynomials in $C_t$ with $\bm{v}_1,\bm{v}_2,\ldots,\bm{v}_{|C_t|}$,
\begin{align}
    F_t = \qty{C_t\bm{v}_i\mid \sqrt{\lambda_i} > \epsilon} \quad\text{and}\quad
    		G_t = \qty{C_t\bm{v}_i\mid \sqrt{\lambda_i} \le \epsilon}.
\end{align}
Append $F_t$ to $F$ and $G_t$ to $G$, respectively.
If $|F_t| = 0$, return $(F, G)$ and terminate. Otherwise, increment $t$ by 1 and go to \textbf{S1}. 
\end{enumerate}
\end{algorithm}

\begin{remark}
At {\bf S1}, Eq.~\eqref{eq:orthogonalization} implies the orthogonal projection of $C_{t}^{\mathrm{pre}}(X)$ to the complementary subspace of the column space of $F_{\le t-1}(X)$:
\begin{align}
C_{t}(X) &= \qty(E_{\abs{C_t}} - F_{\le t-1}(X)F_{\le t-1}(X)^{\dagger})C_{t}^{\mathrm{pre}}(X).    
\end{align} 
We will see that the evaluation vector of any polynomial of degree up to $t-1$ belongs to $\Span{F_{\le t-1}(X)}$, and the orthogonal projection gives us degree-$t$ candidate polynomials whose evaluation vectors are orthogonal to those of any lower-degree polynomials.
\end{remark}

\begin{remark}\label{rem:eov}
At {\bf S3}, a polynomial $C_t\bm{v}_i$ is classified as an $\epsilon$-approximately vanishing polynomial if $\sqrt{\lambda_i}\le \epsilon$, because $\sqrt{\lambda_i}$ equals the extent of vanishing of $C_t\bm{v}_i$ for $X$, i.e.,
\begin{align}
    \norm{(C_{t}\bm{v}_{i})(X)}_2  =\sqrt{\bm{v}_{i}^{\top}C_{t}(X)^{\top}C_{t}(X)\bm{v}_{i}}=\sqrt{\lambda_{i}}.
\end{align}
\end{remark}

\begin{example}
    Let $X = \{ (1.0, 1.0), (0.1, 0), (-1.0, -1.0)\} \subset \bR^2$ and $\epsilon = 0.1$. We consider a simplified case with $\mathfrak{n}(C_t)^{\top}\mathfrak{n}(C_t) = E_{|C_t|}$. The normalized VCA for $(X, \epsilon)$ runs as follows.\footnote{The values are rounded for readability.} First, we initialize $F = \{1\}$, $G = \{\,\} \subset \bR[x, y]$, and $t=1$.
\begin{enumerate}
\item[{\bf S1}] Since $t=1$, we set $C_1^{\mathrm{pre}}=\{x, y\}$. 
 Then, we have
\begin{align}
     C_{1} &= \{x, y\} - \{1\} \mqty(1.00 \\ 1.00 \\ 1.00)^{\dagger}\mqty(1.00 & 1.00 \\ 0.10 & 0.00 \\ -1.00 & -1.00) = \{x - 0.033, y\}.
 \end{align}
 \item[{\bf S2}] From the following generalized eigenvalue problem, 
\begin{align}
    \mqty(0.96 & 1.00 \\ 0.06 & 0.00 \\ -1.03 & -1.00)^{\top}\mqty(0.96 & 1.00 \\ 0.06 & 0.00 \\ -1.03 & -1.00)V = V\Lambda,
\end{align}
we have generalized eigenvalues $\lambda_1 = 4.00, \lambda_2 = 3.33 \times 10^{-3}$ and the corresponding generalized eigenvectors $\bm{v}_1 = (-0.707, -0.706)^{\top}, \bm{v}_2 = (0.706, -0.707)^{\top}$.
\item[{\bf S3}] As $\lambda_1 > \epsilon > \lambda_2$, we have 
\begin{align}
    F_1 = \qty{C_1\bm{v}_1} = \qty{f_1 = -0.707x - 0.706y + 0.0235} \quad\text{and}\quad
    		G_1 = \qty{C_1\bm{v}_2} = \qty{g_1 = 0.706x - 0.707y - 0.023}.
\end{align}
Since $|F_1| \ne 0$, we continue with $t=2$. 
\item[{\bf S1}] As $t=2 > 1$, we construct 
$C_2^{\mathrm{pre}}=\{f_1^2\} = \{(-0.707x - 0.706y + 0.0235)^2\}$. 
 Then, we have
\begin{align}
     C_{2} &= \{(-0.707x - 0.706y + 0.0235)^2\} - \{1, -0.707x - 0.706y + 0.0235\} 
     \mqty(1.00 & -1.390\\ 1.00 &-0.047 \\ 1.00 & 1.437)^{\dagger}\mqty(1.933 \\ 0.002 \\ 2.067) \\
     &= \{0.5008x^2 + 0.9999xy + 0.01663x + 0.4991y^2 + 0.01661y - 1.335\}.
 \end{align}
 \item[{\bf S2}] From the following generalized eigenvalue problem, 
\begin{align}
    \mqty(0.697 & -1.328 & 0.631)\mqty(0.697 \\ -1.328 \\ 0.631)V = V\Lambda,
\end{align}
we have generalized eigenvalues $\lambda = 2.6511$ and the corresponding generalized eigenvector $\bm{v} = (1.0)$.
\item[{\bf S3}] As $\lambda > \epsilon$, we have 
\begin{align}
    F_2 = \qty{f_2 = 0.5008x^2 + 0.9999xy + 0.01663x + 0.4991y^2 + 0.01661y - 1.335}\quad\text{and}\quad
    		G_2 = \qty{\,}.
\end{align}
Since $|F_2| \ne 0$, we continue with $t=3$.
\item[{\bf S1}] While we construct 
$C_3^{\mathrm{pre}}=\{f_1f_2\}$ and compute $C_{3} = \{f_1f_2\} - \{1, f_1, f_2\} 
     \mqty(1(X) & f_1(X) & f_2(X))^{\dagger}(f_1f_2)(X)$, note that
\begin{align}
     C_3(X) = \qty(E_3 - \mqty(1(X) & f_1(X) & f_2(X))\mqty(1(X) & f_1(X) & f_2(X))^{\dagger})(f_1f_2)(X) = (0, 0, 0)^{\top}
 \end{align}
 because this projects $(f_1f_2)(X)$ to the subspace that is orthogonal to the column space of a full-rank matrix $\mqty(1(X) & f_1(X) & f_2(X))$. Importantly, $C_3\ne\{0\}$ as the total degree of $f_1f_2$ is strictly greater than that of $1, f_1, f_2$.
 \item[{\bf S2}] By solving $C_3(X)^{\top}C_3(X)V = V\Lambda$, we have generalized eigenvalues $\lambda =0$ and the corresponding generalized eigenvector $\bm{v} = (1.0)$.
\item[{\bf S3}] As $\lambda \le \epsilon$, we have $F_3 = \qty{\, }, G_3 = C_3 = \qty{g_2}$.
Since $|F_3| = 0$, we terminate with $F = \{1, f_1, f_2\}$ and $G=\{g_1, g_2\}$.
\end{enumerate}
\end{example}

\begin{theorem}(\cite{livni2013vanishing}; rearranged and rephrased)\label{thm:livni}
Let $X\subset\bR^n$ be a finite set of points and $\epsilon\ge 0$. The normalized VCA with $\mathfrak{n}(C_t)^{\top}\mathfrak{n}(C_t) = E_{|C_t|}$ (i.e., the original unnormalized VCA\footnote{To be precise, the original VCA rescales the polynomials in $F_t$  at the end of the step $\textbf{S3}$. to have unit evaluation vector norm}) runs with $(X, \epsilon)$ returns sets of polynomials $F = \{f_1, \ldots, f_r\}, G = \{g_1, \ldots, g_s\}\subset \ring$ satisfying the following.
\begin{enumerate}
    \item The non-constant polynomials of $F$ are all $\epsilon$-nonvanishing for $X$. The polynomials of $G$ are all $\epsilon$-approximately vanishing for $X$.
    \item The evaluation vectors $f_1(X), \ldots, f_r(X)$ are mutually orthogonal. Further, $r \le |X|$.
    \item If $\epsilon = 0$, $r = |X|$, $\Span{F(X)} = \bR^n$, and $\ideal{G} = \videal{X}$. Further, any polynomial $h\in\ring$ can be represented as $h = f + g$ with some $f\in\mathrm{span}(F_{\le \degree{h}})$ and $g\in\langle G_{\le \degree{h}}\rangle$. Particularly, $g\in\videal{X}$ implies $g\in\langle G_{\le \degree{g}}\rangle$. 
\end{enumerate}
\end{theorem}

\subsection{Spurious vanishing problem and normalization}
Let $\alpha \cdot g$ be a non-zero polynomial $g \in \ring$ scaled by $\alpha > 0$, and let $X\subset\bR^n$ be a finite subset of points. As the extent of vanishing $\norm{(\alpha\cdot g)}_2$ can be arbitrarily controlled by $\alpha$, we need normalization of polynomial for fair measurement of the approximate vanishing of a given polynomial. 

In computational algebra, coefficient normalization has been naturally used. However, term-agnostic algorithms such as VCA do not equip normalization. Consequently, they have had a theoretical caveat that their output approximately vanishing polynomials may not be approximately vanishing after normalization, and some of the output nonvanishing polynomials may become approximately vanishing instead. \cite{kera2019spurious} addressed this problem, the \textit{spurious vanishing problem}, and proposed to normalize a polynomial (say, $h \in \cP $) using a linear mapping $\mathfrak{n}: \cP \to \bR^{\ell}$  as $\norm{\mathfrak{n}(h)}_2 = 1$. For example, $\mathfrak{n}(h)$ can be defined as the coefficient vector of $g$. 
\begin{definition} \label{def:normalization-mapping}
Let $\mathfrak{n}: \ring \to \bR^{\ell}$ be a linear mapping. 
The \textbf{normalization vector} of non-zero polynomial $h \in \cP$ refers to $\mathfrak{n}(h)$.
If $\norm{\mathfrak{n}(h)}_2 = 1$, $h$ is said to be \textbf{$\mathfrak{n}$-normalized}.
Let $\mathcal{A}$ be a basis construction algorithm that receives a finite set of points $X\subset\mathbb{R}^n$ and returns $F,G\subset\ring$ such that $\Span{F(X)} = \mathbb{R}^{m}$ and $\videal{X} = \ideal{G}$. 
If the following hold, $\mathfrak{n}$ is called a \textbf{valid normalization mapping} for $\mathcal{A}$. 
\begin{itemize}
\item For any $f\in F$, $\norm{\mathfrak{n}(f)}_2 = 0$ implies $f(X) \in \Span{F_{\le \degree{f}-1}(X)}$. 
\item For any $g\in G$, $\norm{\mathfrak{n}(g)}_2 = 0$ implies $g \in \ideal{G_{\le \degree{g}-1}}$.
\end{itemize}
Let $H = \{h_1,\ldots, h_r\} \subset \ring$.
With a slight abuse of notation, we reuse $\mathfrak{n}$ to define \textbf{normalization matrix} of $H$ by $\mathfrak{n}(H) =\mqty(\mathfrak{n}(h_2) & \mathfrak{n}(h_2) & \cdots & \mathfrak{n}(h_r))\in\mathbb{R}^{\ell\times r}$.
\end{definition}
Given a finite set of polynomials $C_t \subset \ring$, we are interested in the combination vector $\bm{v} \in \bR^{|C_t|}$ that gives us a polynomial $g = C_t\bm{v}$ that is normalized, i.e., with respect to $\norm{\mathfrak{n}(g)}_2 = 1$, and best vanishing for $X$.
\begin{align}
    \min_{\bm{v} \in \bR^{|C_t|}} \norm{g(X)}_2^2, \quad\text{s.t.}\ \ g = C_t\bm{v},\ \norm{\mathfrak{n}(g)}_2^2 = 1.
\end{align}
A simple application of the Lagrange multiplier shows that the generalized eigenvector of the smallest generalized eigenvalue of 
\begin{align}
    C_t(X)^{\top}C_t(X) = \lambda_{\min} \mathfrak{n}(g)^{\top}\mathfrak{n}(g)\bm{v}_{\min}
\end{align}
offers such $\bm{v}$. The normalized VCA computes all the generalized eigenvectors in Eq.~\eqref{eq:gep} to obtain from best vanishing polynomials to worst vanishing ones.

\begin{theorem}[\cite{kera2019spurious}; rephrased]\label{thm:basis}
Let $\mathfrak{n}$ be a valid normalization mapping for the normalized VCA. Then, Theorem~\ref{thm:livni} holds for the normalized VCA. Further, the output polynomials are all $\mathfrak{n}$-normalized.
\end{theorem}

Coefficient normalization uses $\mathfrak{n}_{\mathrm{c}}$ that returns the coefficient vector of polynomial, and $\mathfrak{n}_{\mathrm{c}}$ is a vanishing normalization mapping for the normalized VCA. As noted in Section~\ref{sec:introduction}, however, it is hard to implement coefficient normalization to VCA because polynomials of $F_t$ and $G_t$ are both sum of polynomials of $C_t$, the polynomials of which is product of polynomials, which are also sum of polynomials of $C_{t-1}$. Computing the coefficient vectors of such nested polynomials is hard via numerical computation. Even if it is possible, these polynomials are generally dense, and thus, the length of coefficient vectors grows with order $\cO\qty(\binom{n+t}{n})$. In order to achieve efficient, numerical, and term-order-free computation---or \textit{monomial-agnostic} computation, we need a new valid normalization mapping that can be efficiently computed in the VCA framework.

\section{Gradient normalization}\label{sec:gradient-normalization}
\begin{definition}
Let $X = \{\points\}\subset\bR^n$ be a finite set of points, and let $h \in \ring$.
The linear mapping  
\begin{align}\label{eq:grad-normalization-mapping}
\mathfrak{n}_{\mathrm{g}, X}: \ring \to \bR^{mn}, \quad h \mapsto 
    \mathfrak{n}_{\mathrm{g}, X}(h) = \mqty( \nabla h(\bm{x}_1)^{\top} & \cdots & \nabla h(\bm{x}_m)^{\top}  )^{\top} \in \bR^{mn}.
\end{align}
is called the \textbf{gradient normalization mapping}. The vector $\mathfrak{n}_{\mathrm{g}, X}(h)$ is called the \textbf{graident normalizatoin vector} of $h$.
\end{definition}

\begin{definition}
Let $X = \{\points\}\subset\bR^n$ be a finite set of points, and let $h \in \ring$. The \textbf{gradient semi-norm} of $h$ is
\begin{align}\label{eq:grad-norm}
	\norm{h}_{\mathrm{g},X} := \frac{1}{\gamma}\norm{\mathfrak{n}_{\mathrm{g}, X}(h)}_2 = \frac{1}{\gamma}\qty(\sum_{i=1}^m\norm{\nabla h(\bm{x}_i)}_2^2)^{1/2},
\end{align}
where $\gamma > 0$. 
\end{definition}
The choice of $\gamma$ may depend on the use case, but $\gamma = 1$ or $\gamma = 1/\sqrt{|X|}$ is reasonable. We use $\gamma=1$ in this paper unless otherwise mentioned. 

The normalization using the gradient semi-norm (i.e., gradient normalization) indeed satisfies $(\alpha \cdot h) / \norm{\alpha\cdot h}_{\mathrm{g}, X} = h / \norm{h}_{\mathrm{g}, X}$, and thus, the spurious vanishing problem can be avoided if $\norm{h}_{\mathrm{g}, X} \ne 0$. We will now show that the output $F, G$ of the normalized VCA can be formed only by the polynomials with positive gradient semi-norm of polynomials. Further, the gradient normalization can be performed numerically and exactly without performing differentiation by exploiting the iterative nature of the normalized VCA. 

\subsection{Validity}
We now show that $\mathfrak{n}_{\mathrm{g},X}$ is a valid normalization mapping for the normalized VCA. To this end, we prepare several lemmas, and then prove a theorem using an induction and the lemmas.

\begin{lemma}\label{lem:poly-decomposition-by-derivatives}
Any non-constant polynomial $g\in\ring$ can be represented as
\begin{align}
    g = \sum_{k=1}^n h_k\frac{\partial g}{\partial x_k} + r,
\end{align}
where $h_1, \ldots, h_n, r\in\ring$ and $\mathrm{deg}_k(r) < \mathrm{deg}_k(g)$ for $k=1,2,\ldots,n$. 
\end{lemma}
\begin{proof}
We present a constructive proof. For simplicity of notation, let $d_k=\mathrm{deg}_k(g)$. At $k=1$, if $d_1 = 0$, we set  $h_1=0$ and proceed to $k=2$. Otherwise, we rearrange $g$ according to the degree with respect to $x_1$ as,
\begin{align}
    g = x_1^{d_1}g_1^{(0)} + x_1^{d_1-1}g_1^{(1)} + \cdots + g_1^{(d_1)},
\end{align}
where $g_{1}^{(\tau)}\in\bR[x_2,\ldots, x_n]_{\le \tau}$ for $\tau=0,1,\ldots,d_1$. Then, we have
\begin{align}
    \frac{\partial g}{\partial x_1} = d_1 x_1^{d_1-1}g_1^{(0)} + (d_1-1)x_1^{d_1-2}g_1^{(1)} + \cdots + 0.
\end{align}
By setting $h_1 = x_1/d_1$, we have
\begin{align}
    g &= h_1\frac{\partial g}{\partial x_1} + \frac{r_1}{d_1},
\end{align}
where $r_1 = x_1^{d_1-1}g_1^{(1)} + 2x_1^{d_1-2}g_1^{(2)} \cdots + d_1g_1^{(d_1)}$. Note that $r_1$ satisfies $\mathrm{deg}_1(r_1)\le d_1-1$ and $\mathrm{deg}_{\ell}(r_1)\le d_{\ell}$ for $\ell\ne 1$. Next, we perform the same procedure for $k=2$ and $r_1$. If $d_2=0$, then set $h_2=0$ and $r_2 = r_1$, and proceed to $k=3$; otherwise, rearrange $r_1$ according to the degree with respect to $x_2$ as
\begin{align}
    r_1 = x_2^{d_2}r_2^{(0)} + x_2^{d_2-1}r_2^{(1)} + \cdots + r_2^{(d_2)},
\end{align}
where $r_{2}^{(\tau)}\in\bR[x_3,\ldots, x_n]_{\le \tau}$. By setting $h_2=x_2/d_2$, we obtain
\begin{align}
    g &= h_1\frac{\partial g}{\partial x_1} + h_2\frac{\partial g}{\partial x_2} + r_2,
\end{align}
where $r_2 = x_2^{d_2-1}r_2^{(1)} + 2x_2^{d_2-2}r_2^{(2)} \cdots + d_2r_2^{(d_2)}$ with $\mathrm{deg}_1(r_2)\le d_1-1$, $\mathrm{deg}_2(r_2)\le d_2-1$ and $\mathrm{deg}_{\ell}(r_2)\le d_{\ell}$ for $\ell\ne 1,2$. Repeat this procedure until $k=n$; then, $r:=r_n$ satisfies $\mathrm{deg}_{\ell}(r)\le d_{\ell}-1$ for all $\ell$.
\end{proof}

The following are the critical lemmas in the induction of a theorem that we will pose later.

\begin{lemma}\label{lem:multiplicity-G}
Let $X$ be a finite set of points in $\bR^n$, and let $G \subset \ring$ be a finite set of polynomials such that $\videal{X}_{\le t} = \ideal{G_{\le t}}$ for any non-negative integer $t$. For any $g \in \videal{X}_{\le t+1}$, the zero gradient semi norm $\norm{\mathfrak{n}_{\mathrm{g},X}(g)}_2 = 0$ implies $g \in \langle G_{\le t} \rangle$.
\end{lemma}
\begin{proof}
From Lemma~\ref{lem:poly-decomposition-by-derivatives}, we can represent $g \in \videal{X}_{\le t+1}$ as $g = \sum_{k=1}^n h_k\frac{\partial g}{\partial x_k} + r$ with some $h_1, \ldots h_n\in \ring$ and $r\in \rringle{t}$. From the assumption, we have $\norm{(\partial g/\partial x_k)(X)}_2=0$ and thus $\partial g/\partial x_k \in \videal{X}$ for $k=1,2,\ldots, n$. Consequently, we have $r = g - \sum_{k=1}^n h_k\frac{\partial g}{\partial x_k}\in \videal{X}$. 
Noting that all $\partial g/\partial x_1, \ldots, \partial g/\partial x_n, r$ are of degree up to $t$, we have $g \in \videal{I}_{\le t} = \ideal{G_{\le t}}$.
\end{proof}

Let $X$ be a finite set of points in $\bR^n$, and let $G \subset \ring$ be a finite set of polynomials such that $\videal{X}_{\le t} = \ideal{G_{\le t}}$ for any non-negative integer $t$. For any $g \in \videal{X}_{\le t+1}$, the zero gradient semi norm $\norm{\mathfrak{n}_{\mathrm{g},X}(g)}_2 = 0$ implies $g \in \langle G_{\le t} \rangle$.

\begin{lemma}\label{lem:multiplicity-F}
Let $X$ be a finite set of points in $\bR^n$, and let $F \subset \ring$ be a finite set of polynomials such that for any non-negative integer $t$ and $f \in \Span{F_{\le t}}$, $f(X) \in \Span{F_{\le t}(X)}$. Then, for any $f \in \Span{F_{\le t+1}}$, $\norm{\mathfrak{n}_{\mathrm{g},X}(f)}_2 = 0$ implies $f(X) \in \Span{F_{\le t}(X)}$.
\end{lemma}
\begin{proof}
From Lemma~\ref{lem:poly-decomposition-by-derivatives}, we can represent $f\in \Span{F_{\le t+1}}$ as $f = \sum_{k=1}^n h_k\frac{\partial f}{\partial x_k} + r$,
where $h_k,r \in \ring$ and $\degree{r}\le t$. 
From the assumption, we have $\norm{(\partial f/\partial x_k)(X)}_2=0$ for $k=1,2,\ldots, n$.
We thus have $f(X) = r(X) \in \Span{F_{\le t}(X)}$.  
\end{proof}

Now, we prove that $\mathfrak{n}_{\mathrm{g},X}$ is a valid normalization mapping for the normalized VCA, and with Theorem~\ref{thm:basis}, it is valid to use $\mathfrak{n}_{\mathrm{g},X}$ for the normalization in the normalized VCA.

\begin{theorem}\label{thm:valid-normalization}
Let $X\subset\mathbb{R}^n$ be a set of points. The map $\mathfrak{n}_{\mathrm{g},X}$ defined by Eq.~\eqref{eq:grad-normalization-mapping} is a valid normalization for any basis computation algorithm that receives $X$ and outputs a degree-restriction compatible pair of the vanishing ideal $\videal{X}$.
\end{theorem}
\begin{proof}
A map $\mathfrak{n}_{\mathrm{g},X}: \cP \to \bR^{mn}$ is a linear map because of the linearity of the gradient operator $\nabla$. From Lemmas~\ref{lem:multiplicity-G} and~\ref{lem:multiplicity-F}, the two requirements for the validity of the normalization map in Definition~\ref{def:normalization-mapping} are satisfied. 
\end{proof}

\subsection{Computation}
In the normalized VCA, we compute the gradient normalization vector of polynomial $h \in \ring$.
\begin{align}
    \mathfrak{n}_{\mathrm{g}, X}(h) = \mqty( \nabla h(\bm{x}_1)^{\top} & \cdots & \nabla h(\bm{x}_m)^{\top}  )^{\top} \in \bR^{mn}.
\end{align}
Unlike coefficient vectors, the length of this vector only linearly grows with respect to $n$. Further, in the normalized VCA, this vector can be computed efficiently and exactly without performing differentiation via vector operations. Indeed, let us consider the degree-$t$ step of the normalized VCA. Assume that the gradient normalization vectors are known for the polynomials computed in the lower degree. Then, for any $h \in F_t \cup G_t$, we have the following representation by construction.
\begin{align}\label{eq:linear-combination-of-Ct}
h = \sum_{c\in C_t^\mathrm{pre}} u_c c + \sum_{f\in F_{\le t-1}}v_f f,
\end{align}
where $u_c, v_f\in\mathbb{R}$. Recall that $c \in C_t^\mathrm{pre}$ is defined as the product $c = p_cq_c$ for some $p_c\in F_1$ and $q_c\in F_{t-1}$. Thus, for any $k\in\{1,\ldots, n\}$ and $\bm{x} \in X$, we have
\begin{align}
    \frac{\partial h}{\partial x_k}(\bm{x}) 
    &= \sum_{c\in C_{t}^{\mathrm{pre}}}u_c \frac{\partial (p_cq_c)}{\partial x_k}(\bm{x})
    + \sum_{f\in F_{\le t-1}}v_f \frac{\partial f}{\partial x_k}(\bm{x}) \\
    &= \sum_{c\in C_{t}^{\mathrm{pre}}}u_c q_c(\bm{x})\frac{\partial p_c}{\partial x_k}(\bm{x}) + \sum_{c\in C_{t}^{\mathrm{pre}}}u_c p_c(\bm{x})\frac{\partial q_c}{\partial x_k}(\bm{x}) + \sum_{f\in F_{\le t-1}}v_f \frac{\partial f}{\partial x_k}(\bm{x}).\label{eq:derivative-chain}
\end{align}
The evaluations $p_c(\bm{x}), q_c(\bm{x})$ have already been calculated at the iterations of the lower degrees. Further, from the assumption, we also know $(\partial p_c/\partial x_k)(\bm{x})$,  $(\partial q_c/\partial x_k)(\bm{x})$, and $(\partial f/\partial x_k)(\bm{x})$. Thus, $(\partial h/\partial x_k)(\bm{x})$ and consequently $\mathfrak{n}_{\mathrm{g},X}(h)$ can be calculated without differentiation. Note that at degree $t=1$, the gradients of the linear polynomials are the generalized eigenvectors obtained at {\bf S2}. 

Gradient normalization is the first polynomial-time method that resolves the spurious vanishing problem in a monomial-agnostic manner. The computation can be further accelerated. For example, one can heuristically select a subset $Z\subset X$ and compute gradients over these points. The subset may be determined using a clustering or coreset selection method. When the dimensionality $n$ is large, a coordinate transformation can be used~(see Appendix~\ref{app:preprocessing}).

\section{Advantages of gradient normalization}\label{sec:advantages-of-gradient-normalization}
Here, we present three advantages of gradient normalization besides its monomial-agonociticy: i) Scaling consistency, ii) fidelity to unperturbed points, and iii)   compact generating set. The first two are advantages over coefficient normalization thanks to the data-dependent nature of gradient normalization, and the third shows for the first time that the normalized VCA outputs a smaller (yet equally expressive) generating set than the unnormalized one does. 

\subsection{Scaling consistency}
Scaling of points $\alpha\cdot X = \qty{\alpha \bm{x}_1, \ldots, \alpha\bm{x}_m}$ is a widely used preprocessing for stable numerical computation. For example, Approximate Vanishing Ideal~(AVI) algorithm~\citep{heldt2009approximate} requires that input points are in the range of $[-1, 1]$, and \cite{wirth2023approximate} analyzed their Oracle Approximate Vansihing Ideal (OAVI) algorithm based on the points are in $[0,1]$. A scaling of points readily meets these conditions. An underlying assumption of the scaling is that it has little effect on basis computation. Here, we show that this is not necessarily true if one uses coefficient normalization, which is the case with the aforementioned studies. We will further show that gradient normalization resolves this issue for normalized VCA; that is, gradient-normalized VCA is \textit{scaling consistent}. We start with the following example. 

\begin{example}\label{exapmle:consistency-toy}
Let $h\in\ring$ be a polynomial with non-zero gradient norm for some $\bm{x}\in\mathbb{R}^n$. Let $\alpha\ne 0$, and $\widehat{h} = \sum_{\tau=0}^{\degree{h}} \alpha^{1-\tau}h^{(\tau)}$, where $h^{(\tau)}$ is the degree-$\tau$ part of $h$. The following holds.
\begin{align}\label{eq:scaling-consistency-toy}
\frac{\alpha h(\bm{x})}{\norm{\nabla h(\bm{x})}_2} &= \frac{\widehat{h}(\alpha\bm{x})}{\norm{\nabla \widehat{h}(\alpha\bm{x})}_2}. 
\end{align}
Notably, while $h$ and $\widehat{h}$ may be both nonlinear, the evaluation vectors of their gradient normalization relate linearly in a sense. 
\end{example}

This example might seem trivial because we define $\widehat{h}$ so that Eq.~\eqref{eq:scaling-consistency-toy} hold. However, it is \textit{not} trivial that basis computation for scaled and non-scaled points gives such a pair. 
Notably, we can prove this is the case for the gradient-normalized VCA. 
First, we examine the relationship of $h$ and $\widehat{h}$ in Example~\ref{exapmle:consistency-toy}. 
\begin{definition}[$(t,\alpha)$-degree-wise identical]\label{def:identical}
Let $\alpha \ne 0$ and let $t$ be an integer. A polynomial $\widehat{h}\in\ring$ is said to be \textbf{$(t,\alpha)$-degree-wise identical} to a polynomial $h\in\ring$ if $\widehat{h} = \sum_{\tau=0}^{\degree{h}} \alpha^{t-\tau}h^{(\tau)}$ and $h^{(\tau)}$ is the degree-$\tau$ part of $h$.
\end{definition}
\begin{example}
	$\widehat{h} = x^2y+4x+8y$ is $(3,2)$-degree-wise identical to $h = x^2y + x+ 2y$. 
\end{example}

\begin{lemma}\label{lem:degree-wise-identical}
Let $H = \{h_1,h_2,\ldots,h_r\}\subset\ring$ and $\widehat{H} = \{\widehat{h}_1,\widehat{h}_2,\ldots,\widehat{h}_r\}\subset\ring$, where $\widehat{h}_i$ is $(t,\alpha)$-degree-wise identical to $h$ for $i=1,2,\ldots,r$. 
\begin{enumerate}
    \item  Let $t'$ be a non-negative integer.
    For any nonzero vectors $\bm{w},\widehat{\bm{w}}\in\mathbb{R}^s$ such that $\widehat{\bm{w}}={\alpha}^{t'}\bm{w}$, a polynomial $\widehat{H}\widehat{\bm{w}}$ is $(t+t',\alpha)$-degree-wise identical to $H\bm{w}$.
    \item Let $X\subset\bR^n$ be a finite set of points. Then, it holds that $\widehat{h}_i(\alpha \cdot X) = {\alpha}^{t}h_i(X)$ and $\mathfrak{n}_{\mathrm{g},\alpha\cdot X}(\widehat{h}_i) = {\alpha}^{t-1}\mathfrak{n}_{\mathrm{g},X}(h_i)$ for $i= 1, \ldots, r$. 
\end{enumerate}
\end{lemma}
\begin{proof}
The claims can be shown by straightforward computation.

(1) Let $d_{\max} = \max_{i \in \{1, \ldots, r\}} \degree{h_i}$, and let $\bm{w} = (w_1, \ldots, w_r)^{\top}$. Let $h_i = \sum_{\tau=0}^{\degree{h_i}} h_i^{\tau}$ with degree-$\tau$ part $h_i^{(\tau)}$ for $i = 1, \ldots, r$. Then, we have
\begin{align}
    \widehat{H}\widehat{\bm{w}} 
    = \alpha^{t'}\sum_{i=1}^r w_i\widehat{h}_i 
    = \alpha^{t'}\sum_{i=1}^r w_i \sum_{\tau=0}^{\degree{h_i}} \alpha^{t-\tau}h_i^{(\tau)}
    = \sum_{\tau=0}^{d_{\max}} \alpha^{(t+t')-\tau}\qty(\sum_{i=1}^r w_i h_i^{(\tau)}).
\end{align}
The inner sum of the right-hand side is the degree-${\tau}$ part of $H\bm{w}$, which concludes the claim.

(2) Let $h = \sum_{\tau}h^{(\tau)}, \widehat{h} = \sum_{\tau}\widehat{h}^{(\tau)}$, where $h^{(\tau)}, \widehat{h}^{(\tau)}$ are the degree-$\tau$ parts of $h, \widehat{h}$, respectively, for $\tau = 0, \ldots, \degree{h}$.
Then, we have
\begin{align}
\widehat{h}^{(\tau)}(\alpha X) 
= {\alpha}^{t-\tau}h^{(\tau)}(\alpha X)  
= {\alpha}^{t-\tau}{\alpha}^{\tau}h^{(\tau)}(X) 
= {\alpha}^t h^{(\tau)}(X).
\end{align}
Similarly, for $k \in \{1, 2, \ldots, n\}$, 
\begin{align}
\frac{\partial\widehat{h}^{(\tau)}}{\partial x_k}(\alpha X) 
= {\alpha}^{t-\tau}\frac{\partial h^{(\tau)}}{\partial x_k}(\alpha X)  
= {\alpha}^{t-\tau} {\alpha}^{\tau-1}\frac{\partial h^{(\tau)}}{\partial x_k}(X)
= {\alpha}^{t-1} \frac{\partial h^{(\tau)}}{\partial x_k}(X),
\end{align}
resulting in $\nabla \widehat{h}^{(\tau)}(\alpha X) = {\alpha}^{t-1}\nabla h^{(\tau)}(X)$.
\end{proof}

Now, we claim that the gradient-normalized VCA is consistent with the input translation and scaling. The proof comes later.
\begin{theorem}\label{thm:consistency}
Let $X \subset \ring$ be a finite set of points. Let $\epsilon \ge 0$ and $\alpha \ne 0$. Apply the gradient-normalized VCA to $(X, \epsilon)$ to obtain $F = \{f_1,\ldots, f_r'\}$ and $G = \{g_1, \ldots, g_s'\}$. Similarly, apply it to to $(\alpha \cdot X, \abs{\alpha} \epsilon)$ to obtain $\widehat{F} = \{\widehat{f}_1,\ldots, \widehat{f}_r\}$ and $\widehat{G} = \{\widehat{g}_1, \ldots, \widehat{g}_s\}$.
By arranging the ordering of polynomials in the sets, the two results $(F, G)$ and $(\widehat{F}, \widehat{G})$ relate as follows. 
\begin{enumerate}
    \item $r = r'$ and $s = s'$. 
    \item For all $i$, $f_i$ and $g_i$ are $(1,\alpha)$-degree-wise identical to $f_i$ and $g_i$, respectively, except for the constant polynomials $f_1, \widehat{f}_{1}$.
\end{enumerate}
\end{theorem}
\begin{remark}
The constant polynomial denoted by $f_0 \ne 0$ in Algorithm~\ref{alg:NVCA} can be defined arbitrarily. If coefficient normalization is used, $f_ 0 = 1$ is a reasonable choice. For gradient normalization, we use 
$f_0=\sum_{\bm{x}\in X}\norm{\bm{x}}_{\infty}/m$, where $\norm{\,\cdot\,}_{\infty}$ denotes the $L_{\infty}$ norm,  so that $\widehat{f}_1$ becomes $(1,\alpha)$-degree identical to $f_1$ in Theorem~\ref{thm:consistency}.
\end{remark}

Theorem~\ref{thm:consistency} argues that with gradient normalization, scaling of the input points does not essentially affect the basis computation's output. 
The output basses have the same configuration regardless of the scaling, and the polynomials obtained with and without scaling relate to each other by $(1, \alpha)$-degree identicality. 

As the translation consistency of VCA is already known, gradient-normalized VCA equips translation- and scaling-consistency, which matches the following intuition: the affine variety that approximately accommodates the non-transformed points should be essentially the same as the one that approximately accommodates the transformed points. This is a reasonable assumption; however, none of the existing basis computation algorithms equip this property. Note that as for scaling, this problem cannot be resolved by a simple preprocessing that scales the points to a fixed range before calculation. This is because the preprocessing can yield a set of points, for which the reconstruction of the bases with the configuration of the true bases becomes impossible. 
This will be theoretically and empirically verified in Proposition~\ref{prop:scaling-inconsistency} and Section~\ref{sec:experiment-scaling}, respectively.

From now on, we will prove Theorem~\ref{thm:consistency}. The proof is done by induction and tracking the two processes of the gradient-normalized VCA, one for $(X, \epsilon)$ and another for $(\alpha\cdot X, \epsilon)$. Assuming the claims hold for the degree up to $t$, we show that the claim also holds for degree $t+1$ by examining the steps \textbf{S1}--\textbf{S3} of Algorithm~\ref{alg:NVCA}.

\begin{lemma}\label{lem:orthogonal-projection}
With the setup in Theorem~\ref{thm:consistency}, consider two processes of the gradient-normalized VCA, one for $(X, \epsilon)$ and another for $(\alpha\cdot X, \abs{\alpha}\epsilon)$. Assumes that the claims hold to some degree $\tau$. 
At \textbf{S1}, let $C_{\tau+1} = \{c_1, \ldots, c_k\}$ and $\widehat{C}_{\tau+1} = \{\widehat{c}_1, \ldots, \widehat{c}_{k'}\}$. Then, $k = k'$, and $\widehat{c}_i$ is $(2,\alpha)$-degree-wise identical to $c_i$ for all $i = 1, \ldots, k$.
\end{lemma}
\begin{proof}
We use the notations in Algorithm~\ref{alg:NVCA} for the first process and put $\widehat{\,\cdot\,}$ on the symbols in the second process. 

The sets of pre-candidate polynomials $C_{\tau+1}^{\mathrm{pre}}, \widehat{C}_{\tau+1}^{\mathrm{pre}}$ are generated from $(F_1,F_{\tau})$ and $(\widehat{F}_1,\widehat{F}_{\tau})$, respectively.
From the assumption, their cardinality is the same, i.e., $\abs{C_{\tau+1}^{\mathrm{pre}}} = \abs{\widehat{C}_{\tau+1}^{\mathrm{pre}}}$. Moreover, 
note that $F_1 = \widehat{F_1}$ are a set of linear polynomials, and any polynomial $q \in F_{\tau}$ is $(1, \alpha)$-degree-wise identical to the corresponding $q\in \widehat{F}_{\tau}$. Thus, any $\widehat{c}^{\mathrm{pre}} = \widehat{p}\widehat{q}$ is degree-wise-($2,\alpha$) identical to the corresponding $c^{\mathrm{pre}} = pq \in C_{\tau+1}^{\mathrm{pre}}$.

Next, the element-wise description of the orthogonal projection Eq.~\eqref{eq:orthogonalization} for $C_{\tau+1}$ and $\widehat{C}_{\tau+1}$ is respectively as follows.
\begin{align}
    c = c^{\mathrm{pre}} - F_{\le \tau}F_{\le \tau}(X)^{\dagger}c^{\mathrm{pre}}(X)\quad \text{ and }\quad
    \widehat{c} = \widehat{c}^{\mathrm{pre}} - \widehat{F}_{\le \tau}\widehat{F}_{\le \tau}(\alpha X)^{\dagger}\widehat{c}^{\mathrm{pre}}(\alpha X).
\end{align}
Let $\bm{w}=F_{\le \tau}(X)^{\dagger}c^{\mathrm{pre}}(X)$ and $\widehat{\bm{w}}=\widehat{F}_{\le \tau}(\alpha X)^{\dagger}\widehat{c}^{\mathrm{pre}}(\alpha X)$. We will now show that $\widehat{\bm{w}} = \alpha \bm{w}$. If this holds, by Lemma~\ref{lem:degree-wise-identical}, each element of $\widehat{F}_{\le \tau}\widehat{\bm{w}}$ becomes degree-wise-($2,k$) identical to the corresponding element of $F_{\le \tau}\bm{w}$. Thus, $\widehat{c}$ is $(2,\alpha)$-degree-wise identical to $c$.

First, note that the column vectors of $\widehat{F}_{\le \tau}(\alpha X)$ are mutually orthogonal by construction since the orthogonal projection makes $\mathrm{span}(\widehat{F}_{t_1}(\alpha X))$ and $\mathrm{span}(\widehat{F}_{t_2}(\alpha X))$ mutually orthogonal for any $t_1\ne t_2$, and the generalized eigenvalue decomposition makes the columns of $\widehat{F}_{t}(\alpha X)$ mutually orthogonal for each $t$. Therefore, 
\begin{align}
    \widehat{D} := \widehat{F}_{\le \tau}(\alpha X)^{\top}\widehat{F}_{\le \tau}(\alpha X)
    ={\alpha}^{2}F_{\le \tau}(X)^{\top}F_{\le \tau}(X)
    =: {\alpha}^{2} D,
\end{align}
where both $\widehat{D}$ and $D$ are diagonal matrices with positive diagonal elements.
Hence, the pseudo-inverse becomes
\begin{align}
\widehat{F}_{\le \tau}(\alpha X)^{\dagger}&=\widehat{D}^{-1}\widehat{F}_{\le \tau}(\alpha X)^{\top}
= ({\alpha}^{-2}D^{-1})(\alpha F_{\le \tau}(X)^{\top})
= {\alpha}^{-1} D^{-1}F_{\le \tau}(X)^{\top} 
= {\alpha}^{-1} F_{\le \tau}(X)^{\dagger},  
\end{align}
and thus, $\widehat{\bm{w}} = \alpha \bm{w}$.
\end{proof}

\begin{lemma}\label{lem:linear-response}
With the setup in Theorem~\ref{thm:consistency}, consider two processes of the gradient-normalized VCA, one for $(X, \epsilon)$ and another for $(\alpha\cdot X, \abs{\alpha}\epsilon)$. Assumes that the claims hold to some degree $\tau$. 
At \textbf{S2} and $\textbf{S3}$, let $F_{\tau+1} = \{f_1, \ldots, f_k\}$, $\widehat{F}_{\tau+1} = \{\widehat{f}_1, \ldots, \widehat{f}_{k'}\}$, $G_{\tau+1} = \{g_1, \ldots, g_l\}$, and $\widehat{G}_{\tau+1} = \{\widehat{g}_1, \ldots, \widehat{g}_{l'}\}$. Then, $k = k'$, $l=l'$, and  $\widehat{f}_i$ and $\widehat{g}_i$ are $(1,\alpha)$-degree-wise identical to $f_i$ and $g_i$, respectively, for each $i$.
\end{lemma}
\begin{proof}
We use the notations in Algorithm~\ref{alg:NVCA} for the first process and put $\widehat{\,\cdot\,}$ on the symbols in the second process. 
From Lemma~\ref{lem:orthogonal-projection}, each $\widehat{c} \in \widehat{C}_{\tau+1}$ is $(2, \alpha)$-degree-wise identical to the corresponding $c\in C_{\tau+1}$. Thus, from Lemma~\ref{lem:degree-wise-identical}, we have $\widehat{C}_{\tau+1}(\alpha X) = {\alpha}^{2}C_{\tau+1}(X)$ and $\gnvec{\alpha X}{\widehat{C}_{\tau+1}} = \alpha \gnvec{X}{C_{\tau+1}}$.
Using these results, we now compare two generalized eigenvalue problems. 
\begin{align}
C_{\tau+1}(X)^{\top}C_{\tau+1}(X)V
    &= \gnvec{X}{C_{\tau+1}}^{\top}\gnvec{X}{C_{\tau+1}}
    V\Lambda, \\ 
    \widehat{C}_{\tau+1}(\alpha X)^{\top}\widehat{C}_{\tau+1}(\alpha X)\widehat{V}
    &= \gnvec{\alpha X}{\widehat{C}_{\tau+1}}^{\top}\gnvec{\alpha X}{\widehat{C}_{\tau+1}}
    \widehat{V}\widehat{\Lambda}.
\end{align}
Let us simplify the notations to $N = \gnvec{X}{C_{\tau+1}}^{\top}\gnvec{X}{C_{\tau+1}}$ and $\widehat{N} = \gnvec{\alpha X}{\widehat{C}_{\tau+1}}^{\top}\gnvec{\alpha X}{\widehat{C}_{\tau+1}}$.
From $\widehat{N} = \alpha^2 N$ and $V^{\top}NV = \widehat{V}^{\top}\widehat{N}\widehat{V} = E_{\abs{C_{\tau+1}}}$, we obtain $\widehat{V} = {\alpha}^{-1}V$. Let $\bm{v}_i$ and $\widehat{\bm{v}}_i$ be the $i$-th columns of $V$ and $\widehat{V}$, respectively. From $\widehat{\bm{v}}_i={\alpha}^{-1}\bm{v}_i$ and Lemma~\ref{lem:degree-wise-identical}, $\widehat{C}_{1}\widehat{\bm{v}}_i$ is $(1,\alpha)$-degree-wise identical to $C_{\tau+1}\bm{v}_i$.
Lastly, note that polynomials in $\widehat{C}_{\tau+1}\widehat{V}$ are classified into $\widehat{F}_{\tau+1}$ or $\widehat{G}_{\tau+1}$ by the threshold $|\alpha| \epsilon$. This leads to the same classification as $F_{\tau+1}$ and $G_{\tau+1}$ by $\epsilon$. Consequently, the size of $F_{\tau+1}$ and $\widehat{F}_{\tau+1}$ is the same, and that of $G_{\tau+1}$ and $\widehat{G}_{\tau+1}$ is also the same. 
\end{proof}

We now prove Theorem~\ref{thm:consistency} using induction and Lemma~\ref{lem:linear-response}.

\begin{proof}[Proof of Theorem~\ref{thm:consistency}.]
In the linear case, i.e., $\tau = 1$, the claim is true by construction. Assume that the claim holds for degree $t = 1, \ldots, \tau$. Then, from Lemma~\ref{lem:linear-response}, the claim also holds for $t = \tau + 1$. 
\end{proof}

Lastly, we show that scaling \textit{inconsistency} of the coefficient-normalized VCA.
\begin{proposition}\label{prop:scaling-inconsistency}
Let $X\subset \bR^n$ be a finite, mean-subtracted set of points. Consider two processes of the coefficient-normalized VCA, one for $(X, \epsilon)$ and another for $(\alpha\cdot X, \epsilon_{\alpha})$, where $\epsilon, \epsilon_{\alpha} \ge 0$ and $\alpha \ne 0$. The former returns $(F, G)$, and the latter returns $(\widehat{F}, \widehat{G})$.
If $F$ contains a degree-1 polynomial, and $G$ contains any approximately (but not exactly) vanishing polynomial of degree 1, then there exists an $\alpha\ne 0$ such that $|G_2| \ne |\widehat{G}_{2}|$.
\end{proposition}
\begin{proof}
We use the notations in Algorithm~\ref{alg:NVCA} for the first process and put $\widehat{\,\cdot\,}$ on the symbols in the second process. 
By assumption, $X$ is mean-subtracted, and thus, the orthogonal projection with a constant vector $F_0(X)$, which is equivalent to mean subtraction, does not change $C^{\text{pre}}_1(X)=X$ (i.e., $C_1 = C_1^{\text{pre}}, \widehat{C}_1 = \widehat{C}_1^{\text{pre}}$), leading to $F_1 = \widehat{F}_1$. Thus, it holds that $C_2^{\text{pre}} =  \widehat{C}_2^{\text{pre}}$. Note that now, $F_1$ and $C_2^{\text{pre}}$ only contain homogeneous polynomials of degree 1 and 2, respectively. Thus, it holds that $F_1(\alpha X) = \alpha F_1(X)$ and $C_2(\alpha X) = \alpha^2C_2(X)$. Additionally, recall that the column vectors of $F_{\le 1}(X) = \mqty(F_0(X) & F_1(X))$ are mutually orthogonal by construction. Thus, it holds that 
    $F_1(X)^{\dagger} = D^{-2}F_1(X)^{\top}$,
where $D$ is the diagonal matrix with the Euclidean norm of each column vector of $F_{\le 1}(X)$ in the diagonal. Similarly, we define $D^{(\alpha)}$ for $\widehat{F}_{1}(X)$. Note that 
\begin{align}
\widehat{F}_1(\alpha X)^{\dagger} 
&= \qty(D^{(\alpha)})^{-2}\widehat{F}_1(\alpha X)^{\top}
= D^{-2}F_1(X)^{\top} 
= F_1(X)^{\dagger}.
\end{align}
Then, with $F_1 = \widehat{F}_1$ and $C_2^{\text{pre}} = \widehat{C}_2^{\text{pre}}$, the orthogonal projection at $t=2$ yields 
\begin{align}
    \widehat{C}_2 
    &= \widehat{C}_2^{\text{pre}} - \widehat{F}_{\le 1}\widehat{F}_{\le 1}(\alpha X)^{\dagger}\widehat{C}_2^{\text{pre}}(\alpha X)
     = C_2^{\text{pre}} - \alpha^2 F_{\le 1}F_{\le 1}( X)^{\dagger}C_2^{\text{pre}}(X). %
\end{align}
We also have 
\begin{align}
    \widehat{C}_2(\alpha X) 
    &= C_2^{\mathrm{pre}}(\alpha X) - \alpha^2 F_{\le 1}F_{\le 1}( X)^{\dagger}C_2^{\text{pre}}(X)
    = \alpha^2 C_2(X). \label{eq:eval-multiplicative}
\end{align}
When $\alpha$ goes to zero,
$\widehat{C}_2$ approaches to $C_2^{\mathrm{pre}}$, whereas $\widehat{C}_2(\alpha X)$, and thus any coefficient-normalized polynomial with support $\widehat{C_2}$, goes to zero in a quadratic order.  By assumption, $G_1 (=\widehat{G}_1)$ contains an approximately vanishing polynomial, the extent of vanishing of which also goes to zero as $\alpha \to 0$ but in a linear order. Therefore, there exists $1 \gg \alpha_0 > 0$ such that all degree-2 polynomials obtained from the generalized eigenvalue problem~Eq.~\eqref{eq:gep} vanish more strictly than those in $\widehat{G}_1$. Thus, one cannot set $\epsilon_{\alpha}$ to maintain $\widehat{G}_1$ and quadratic (or higher-degree) nonvanishing polynomials simultaneously. 
\end{proof}

\subsection{Fidelity to unperturbed points}
Approximately vanishing for a set of perturbed points $X\subset\mathbb{R}^n$ does not necessarily imply approximately vanishing for the set of unperturbed points $X^*$. The gap in the extent of vanishing between perturbed and unperturbed points can be significant when the polynomial is coefficient-normalized. 

\begin{example}\label{example:cons-of-coefficient-normalization}
Let a set of a single one-dimensional point $X = \{p\}\subset\mathbb{R}$. Let $p^* = p - \delta$ be an unperturbed point of $p$, where $\delta \in \bR$ is the perturbation.  
Let $g= (x-p)x^5 / \sqrt{1+p^2} \in \bR[x]$. The polynomial $g$ is coefficient-normalized and vanishing for $X$. However, $g(p^{*}) = \delta (p^*)^{5} / \sqrt{1+(p^*)^2} = O(\delta \cdot (p^*)^3)$. The effect of perturbation can be arbitrarily large by increasing the value of $p$. 
\end{example}

In contrast, for a gradient-normalized polynomial, approximately vanishing for $X$ mildly implies approximately vanishing for $X^*$, where the gap of the two extents of vanishing increases linearly according to the perturbation magnitude. 

\begin{example}\label{example:pros-of-gradient-normalization}
Let us consider the same setting in Example~\ref{example:cons-of-coefficient-normalization}. If $g$ is gradient-normalized, i.e., $\widetilde{g} := g / \abs{(dg/dx)(p)} = (x-p)$, then its evaluation at $p^*$ is $\widetilde{g}(p^*) = \delta/5 = O(\delta)$, which is not dependent on the value of $p^*$. 
\end{example}

This difference between coefficient normalization and gradient normalization arises from the fact that the latter is a data-dependent normalization. 
Although Examples~\ref{example:cons-of-coefficient-normalization} and~\ref{example:pros-of-gradient-normalization} cannot be directly generalized to multivariate and multiple-point cases, we can still prove a similar statement. The following proposition argues that when the perturbation is sufficiently small, the extent of vanishing at unperturbed points is also small. Moreover, it is bounded only by the largest perturbation and not by the norm of the points, which is not the case with coefficient normalization as shown in Example~\ref{example:cons-of-coefficient-normalization}. 

\begin{proposition}\label{prop:perturbation}
Let $X^* = \{\bm{x}_1^*, \ldots, \bm{x}_m^*\}\subset \mathbb{R}^n$ be a finite set of points, and let $X = \{\bm{x}_1, \ldots, \bm{x}_m\}$ be its perturbed version.
Let $g \in \ring$ be a gradient-normalized polynomial. Then,
$\norm{g(X)-g(X^*)}_2 \le \norm{\bm{n}_\mathrm{max}}_2 + o(\norm{\bm{n}_\mathrm{max}}_2)$, where $\bm{n}_\mathrm{max} =  \max_{i\in \{1,\ldots, m\}} \norm{\bm{n}_{i}}_2$, and $o(\,\cdot\,)$ is the Landau's small o.
\end{proposition}
\begin{proof}
 Let us denote the perturbation by $\bm{n}_{i} = \bm{x}_i^* - \bm{x}_i$ for $i=1, \ldots, m$. From the Taylor expansion, 
 \begin{align}
     \norm{g(X) - g(X^*)}_2
     &= \sqrt{\sum_{i=1}^m\qty(\nabla g(\bm{x}_i)\bm{n}_{i} + o(\norm{\bm{n}_{i}}_2^2))}^2
     \le \sqrt{\sum_{i=1}^m(\nabla g(\bm{x}_i)\bm{n}_{i})^2} + o(\norm{\bm{n}_{i}}_2). 
 \end{align}
By the Cauchy--Schwarz inequality and $\norm{\gnvec{X}{g}}_2=1$, the second term becomes $\sqrt{\sum_{\bm{x}\in X}(\nabla g(\bm{x}_i)\bm{n}_i)^2}  \le \max_{i\in \{1,\ldots, m\}} \norm{\bm{n}_{i}}_2$.
\end{proof}

This result implies that the magnitude of the difference between the evaluation vectors of perturbed and unperturbed points is proportional to the magnitude of the perturbations when the perturbations are small. 
This behavior is significantly important in practical scenarios where one must empirically adjust a proper $\epsilon$. The prior on the perturbation is typically a Gaussian distribution $\mathcal{N}(\bm{0}, \epsilon^2 E_n)$, which assumes the Euclidean distance from a perturbed point to the unperturbed point is at most $\epsilon$ with a moderate probability. This is a soft thresholding based on the geometrical distance, which can be empirically estimated, e.g., by repeatedly measuring data points. However, in the approximate computation of vanishing ideals, the threshold $\epsilon$ is set for the magnitude of the evaluation values of the polynomials at points. Such an algebraic distance is difficult to estimate without knowing the polynomials. Proposition~\ref{prop:perturbation} provides a (locally) linear relation between the evaluation of polynomials and the geometrical distance. Thus, it is reasonable to set $\epsilon$ solely based on the estimated perturbation magnitude. 

The geometrical distance of approximately vanishing polynomials was also considered in~\citep{fassino2013simple}. As in our analysis, they used the first-order approximation. The main differences between their work and our work are: (i) they focus on finding a single polynomial rather than a basis, (ii) their upper bound is output-sensitive, and (iii) there is a possibility that no polynomial is found with their algorithm. This is because, with coefficient normalization, we cannot ensure a moderate gradient of polynomials around the given points. 

\subsection{Smaller generating set}
The unnormalized VCA may contain spurious vanishing polynomials. A most evident example is the set (or list) of degree-2 pre-candidate polynomials $C_2^{\mathrm{pre}} = \qty{pq \mid p,q \in F_1}$, which, and consequently $C_2$, contain duplicates by construction. Thus, the best vanishing linear combination is the difference of the duplicates, which theoretically gives exact vanishing. Normalization excludes the zero polynomial as the generalized eigenvalue of the corresponding generalized eigenvector is infinity. Similarly, the polynomials that are approximately vanishing only because of their small coefficients (i.e., spurious vanishing polynomials) become nonvanishing once normalized.

In the original study of VCA~\citep{livni2013vanishing}, the size of $G$ was only loosely upper bounded as $|G| \le m^2\min\{m, n\}$ partly because it could include spuriously vanishing polynomials. Recently, \cite{wirth2022conditional} proposed a monomial-\textit{aware} basis computation algorithm, which provides a more compact set of generators than VCA. The upper bound associated with their method is $|G| \le mn$. We now show that the normalized VCA attains $|G| \le (m-n)n$. 
\begin{lemma}\label{lem:ub-number-of-nonvanishing}
Let $(F, G)$ be the output of the normalized VCA given $X\subset\mathbb{R}^n$ and $\epsilon \ge 0$. It holds that $|F_{\le t}| = \sum_{\tau=0}^{t}|F_{\tau}| \le \binom{n+t}{n}$ for any non-negative integer $t$. When the equality holds, $\abs{G_t} = 0$.
\end{lemma}
\begin{proof}
Polynomials of degree at most $t$ are linear combinations of at most $\binom{n+t}{n}$ terms. We show that either normalized or unnormalized VCA outputs a set of nonvanishing polynomials $F = \{f_1, \ldots, f_r\}$ with linearly independent coefficient vectors. If so, we immediately have $|F_{\le t}| \le \binom{n+t}{n}$. Let $\Phi(X)\in \bR^{|X| \times \binom{n+t}{n}}$ be the matrix whose column vectors are evaluation vectors of all terms of degree at most $t$. Then, we have $F(X) = \Phi(X)\mqty(\bm{u}_1 & \cdots & \bm{u}_r)$ from coefficient vector $\bm{u}_1, \ldots, \bm{u}_r$. As the rank of $F(X)$ is a full rank matrix, $\bm{u}_1, \ldots, \bm{u}_r$ should be linearly independent. 
When $|F_{\le t}| = \binom{n+t}{n}$, any degree-$t$ normalized polynomial is a normalized $\epsilon$-nonvanishing polynomial. Thus, $G_t$ is an empty set.
This concludes the proof.
\end{proof}
\begin{proposition}\label{prop:ub-basis-size}
Let $(F, G)$ be the output of the coefficient- or gradient-normalized VCA given $X\subset\mathbb{R}^n$ and $\epsilon \ge 0$. Then, we have $|G| \le (m-n)n$. 
\end{proposition}
\begin{proof}
We consider the worst-case scenario and show that the size of $G$ is yet bounded by $(m-n)n$. Let $t$ be a non-negative integer. By construction, we have $|C_t| = |F_{t-1}||F_1|$ and $|C_t| \ge |F_t| + |G_t|$. The latter holds because $F_t, G_t$ are constructed using the generalized eigenvectors of square matrices in $\bR^{|C_t| \times |C_t|}$.
The best strategy to grow the size of final $G$ is to put as many polynomials as possible into $F_t$ at each degree $t$ so that the number of candidate polynomials in the next degree, i.e., $|C_{t+1}|$, becomes large as possible. Then, at the degree $T$ where the algorithm terminates, all the generated polynomials are put into $G_T$.
Noting that $|F_1|\le n$ and $|F|\le m$ always hold, this strategy leads to $|F_{\le t}| = \binom{n+t}{n}$ for $t = 0, \ldots, T-1$. Thus, $|G_T| \le |C_{T}| = |F_1||F_{T-1}| \le n(m-n)$. From Lemma~\ref{lem:ub-number-of-nonvanishing} and our strategy, $|G_t| = 0$ for $t = 0, \ldots, T-1$.
We obtain $G = G_T$, concluding our proof.

\end{proof}

\section{Further applications of gradient information}\label{sec:futher-applications}
In addition to normalization, we can further exploit the gradient information to boost the monomial-agnostic basis computation of vanishing ideals. We will now introduce two examples: (i) a basis-size reduction method and (ii) a method to handle positive-dimensional ideals.

\subsection{Removal of redundant basis polynomials}\label{sec:removal-of-redundant basis polynomials}
Let $X \subset \ring$ be a finite set of points, and let $G \subset \cP$ be a finite set such that $\videal{X} = \ideal{G}$. The normalized VCA with $\epsilon = 0$ gives us to compute such $G$. Proposition~\ref{prop:ub-basis-size} implies that normalization leads to a smaller size of $G$.  However, we observe that the output $G$ is still sometimes redundant; that is, for some $g \in G$, $\ideal{G} = \ideal{G \setminus \{g\}}$. Note that this is also the case with border bases, which are computed by computer-algebraic basis computation (e.g., the AVI algorithm). 
Here, we provide a numerical approach to remove such redundant polynomials to reduce the basis size. This is useful in some applications; for example, bases of vanishing ideals have been used to construct feature vectors for classification tasks~\citep{livni2013vanishing,zhao2014hand,kera2018approximate,wang2019polynomial}. The reduced bases will provide compact feature vectors, which enhance efficiency and interpretability. 

To determine redundant polynomials, a standard approach in computational algebra is to divide $g$ by the Gr\"obner basis of $\ideal{G\setminus \{g\}}$ and solve an ideal membership problem. However, this approach requires the expensive and numerically unstable computation of the Gr\"obner basis, and it is also not compatible with the approximate settings that are of our interest. 
We want to handle the redundancy in an efficient, noise-tolerant, and monomial-agnostic manner. To this end, we again resort to the gradient of the polynomials. The following conjecture argues that $g$ can be considered redundant if for any $\bm{x}\in X$, the gradient $\nabla g(\bm{x})$ belongs to the column space that is spanned by the gradients of the polynomials in $G_{\le \degree{g}-1}$.

\begin{conjecture}\label{conj:multiplicity-theorem}
Let $X \subset \bR^n$ be a finite set of points, and let $G \subset \ring$ be a finite generating set of $\videal{X}$. Suppose that for any $g \in \videal{I}$, we have $g \in \ideal{G_{\le \degree{g}}}$.
Then, it holds that $g\in \langle G_{\le \degree{g}-1}\rangle$ for some $g \in G$ if and only if for any $\bm{x}\in X$, 
\begin{align}\label{eq:span-by-other-gradient}
    \nabla g(\bm{x}) &= \sum_{g^{\prime}\in  G_{\le \degree{g}-1}} u_{g^{\prime},\bm{x}} \nabla g^{\prime}(\bm{x}),
\end{align}
for some $u_{g^{\prime},\bm{x}}\in\mathbb{R}$.
\end{conjecture}
The sufficiency can be readily proved.
From the assumption $g\in \langle G_{\le \degree{g}-1}\rangle$, we can represent $g$ as $g=\sum_{g^{\prime}\in G_{\le \degree{g}-1}}g^{\prime}h_{g^{\prime}}$ for some $\{h_{g^{\prime}}\}\subset \mathcal{P}_n$. Thus, 
\begin{align}\label{eq:grad-expansion}
    \nabla g(\bm{x}) &= \sum_{g^{\prime}\in G_{\le \degree{g}-1}} h_{g^{\prime}}(\bm{x})\nabla g^{\prime}(\bm{x})
    +g^{\prime}(\bm{x})\nabla h_{g^{\prime}}(\bm{x})
    = \sum_{g^{\prime}\in G_{\le \degree{g}-1}} h_{g^{\prime}}(\bm{x})\nabla g^{\prime}(\bm{x}).
\end{align}
We used $g^{\prime}(\bm{x})=0$ in the last equality. 

From the sufficiency, we can remove all the redundant polynomials from a basis by checking whether Eq.~\eqref{eq:span-by-other-gradient} holds. We may accidentally remove some polynomials that are not redundant because the necessity remains unproven.  Conceptually, the necessity implies that the global (symbolic) relation $g\in \langle G_{\le \degree{g}-1}\rangle$ can be inferred from the local relation Eq.~\eqref{eq:span-by-other-gradient} at finitely many points $X$. This is not true for general $g$ and $G_{\le \degree{g}-1}$. However, $g$ and $G_{\le \degree{g}-1}$ are constructed in a very restrictive way; hence, we suspect that our conjecture may be true. 

We also support the validity of using Conjecture~\ref{conj:multiplicity-theorem} from another perspective. When Eq.~\eqref{eq:span-by-other-gradient} holds, this indicates that with the basis polynomials of lower degrees, one can generate a polynomial $\widehat{g}$ that takes the same value and gradient as $g$ at all the given points; in other words, $\widehat{g}$ behaves identically to $g$ up to the first order for all the points. According to the essence of the vanishing ideal computation---identifying a polynomial only by its behavior at given points---it is reasonable to consider $g$ as ``redundant" for practical use. 

Now, we describe the removal procedure of redundant polynomials based on Conjecture~\ref{conj:multiplicity-theorem}. Given $g$ and $G_{\le \degree{g}-1} = \{g_1, \ldots, g_k\}$, we solve the following least-squares problem for each $\bm{x}\in X$: 
\begin{align}
    \min _{\bm{v}\in\mathbb{R}^{|G_{\le \degree{g}-1}|}}
    \norm{\nabla g(\bm{x}) - \mqty(\nabla g_1(\bm{x}) & \cdots & \nabla g_k(\bm{x}))^{\top} \bm{v}}.
\end{align}
Let $\nabla G(\bm{x}) := \mqty(\nabla g_1(\bm{x}) & \cdots & \nabla g_k(\bm{x})) \in \bR^{n\times k}$.
The problem has a closed-form solution $\bm{v}^{\top}=\nabla g(\bm{x})^{\top}\qty(\nabla G_{\le \degree{g}-1}(\bm{x}))^{\dagger}$. If the residual error equals zero at all the points in $X$, then $g$ is removed as a redundant polynomial. In the approximately vanishing case ($\epsilon > 0$), a threshold for the residual error needs to be set. 
Note that when $G$ is computed by the gradient-normalized VCA, for any $g \in G_{t}$ (not $G_{\le t}$), it holds that $g \notin \ideal{G_t - \{g\}}$ by construction. Thus, it is not necessary to consider the reduction within the same degree. 

\begin{figure*}
\includegraphics[width=\linewidth]{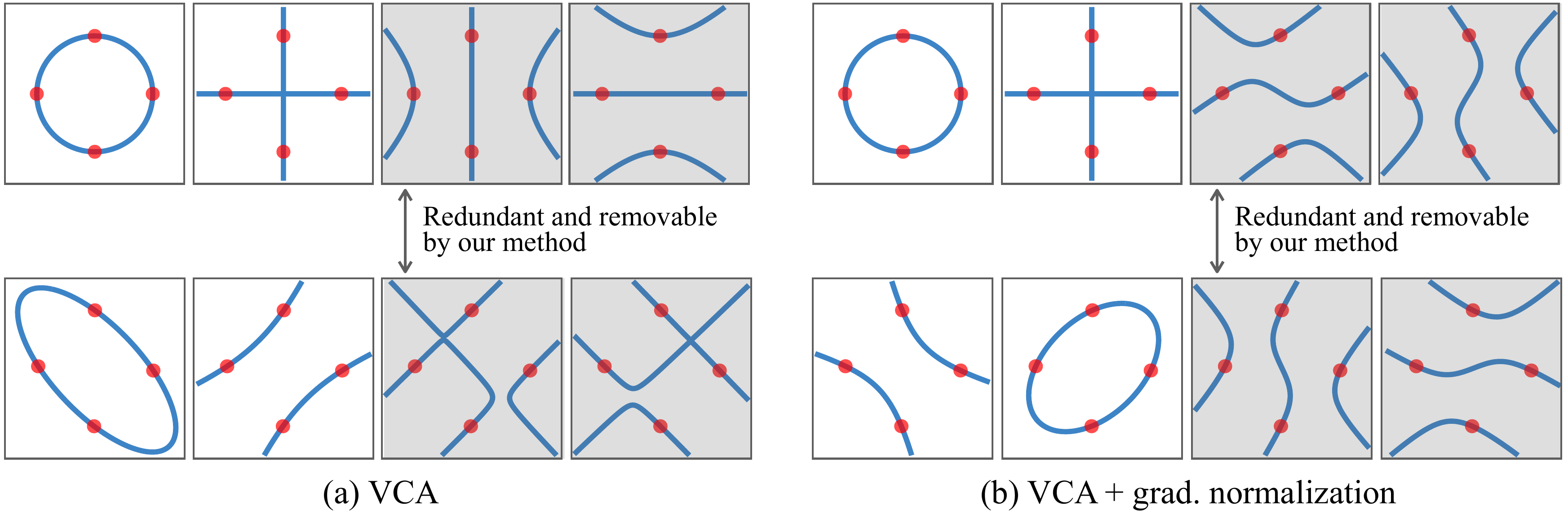}\caption{Contour plots of vanishing polynomials obtained by (a) VCA and (b) VCA with gradient normalization (right panel) on \{(0,0), (1, 0), (0, 1), (1,1)\} (top row) and the perturbed one (bottom row). Both bases contain redundant basis polynomials (shaded plots), which can be efficiently removed by the proposed method based on Conjecture~\ref{conj:multiplicity-theorem}.
\label{fig:result-redundancy}}
\end{figure*}
Lastly, we provide an example where the redundancy of a basis can be removed by our basis reduction method. We consider the vanishing ideal of $X=\{(1,0),(0,1),(-1,0),(0,-1)\}$ with no perturbation, where the exact Gr\"obner basis and polynomial division can be computed for verification. 
As shown in Figure~\ref{fig:result-redundancy}, the VCA and gradient-normalized VCA bases consist of five and four vanishing polynomials, respectively. The bases share $g_1=x^2+y^2-1$  and $g_2=xy$ (the constant scale is ignored). A simple calculation using the Gr\"obner basis of $\{g_1,g_2\}$ reveals that the other polynomials in each basis can be generated by $\{g_1,g_2\}$. Our basis reduction method successfully reduced both bases to $\{g_1,g_2\}$.

\subsection{For positive-dimensional ideals}\label{sec:positive-ideal}
The vanishing ideal of a finite set of points is zero-dimensional. However, in several applications, data points are related to positive-dimensional ideals. Therefore, the algorithms for the basis computation of vanishing ideals sometimes over-specify the affine variety of the data, and they do not generalize to unobserved data well.  of For better generalization, a typical regularization method restricts the maximum degree of basis polynomials or the basis size. However, it is difficult to set these parameters properly because the potential range of the degrees and the number of vanishing polynomials are unknown. Instead of degree or basis-size restriction, we propose using the affine variety dimension for the regularization. The gradient information is again exploited. 

\begin{definition}[Dimension of affine variety]
Let $G = \{g_1,\ldots, g_s\}\subset \ring$ be a finite set of polynomials. 
Let $\mathcal{V}\subset\mathbb{R}^n$ be the affine variety determined as the zero set of $G$. 
Dimension $\mathrm{dim}(\mathcal{V})$ of $\mathcal{V}$ is defined as $\mathrm{dim}(\mathcal{V}) = n - \min_{\bm{x}\in\mathcal{V}^*} \mathrm{rank}(\nabla G(\bm{x}))$,
where $\mathcal{V}^*$ is the set of nonsingular points of $\mathcal{V}$, and $\nabla G(\bm{x}) = \mqty(\nabla g_1(\bm{x}) & \cdots & \nabla g_s(\bm{x}))$. 
\end{definition}
In practice, we must estimate the dimension of the affine variety from finite sample points $X\subset\mathbb{R}^n$. 
Particularly, the following two numbers are of our interest. 
\begin{align}
    d_{\mathrm{max}} &= n - \min_{\bm{x}\in X,  \nabla G(\bm{x}) \ne \bm{0}} \mathrm{rank}(\nabla G(\bm{x}))\quad \text{ and }\quad
    d_{\mathrm{min}} = n - \max_{\bm{x}\in X} \mathrm{rank}(\nabla G(\bm{x})),
\end{align}
Consider the normalized VCA for $X$ and $\epsilon = 0$. Without any restriction, the computed basis $G$ generates a zero-dimensional ideal, i.e., $\videal{X}$. We can take $d_{\max}$ or $d_{\min}$ into account in the basis computation.
The normalized VCA process extends $G$ by appending $G_t$ at each degree $t$. Let $G_{\le t} = \{g_1, \ldots, g_{k_t}\}$. Then, the dimension of the subspace spanned by $\nabla g_1(\bm{x}), \ldots, \nabla g_{k_t}(\bm{x})$ (i.e., the codimension of the tangent space) at each $\bm{x} \in X$ monotonically increases along $t= 1, 2, \ldots$.
The basis computation terminates when this codimension exceeds $n-d_{\max}$ at all the points in $X$. In contrast, by restricting $d_{\min}$, the basis computation terminates when the codimension exceeds $n-d_{\min}$ at some point. 

The significant advantage of using the codimension of the tangent space for regularization rather than the degree or the basis size is that the potential range of the codimension is evident.
If one deals with data points in $\mathbb{R}^n$, then the dimension ranges from $0$ to $n$. When $d_{\max} = 0$ or $d_{\min} = 0$ is used, the full basis computation is performed. This approach is also robust against the redundancy of a basis because, as discussed in Section~\ref{sec:removal-of-redundant basis polynomials}, the redundant basis polynomials do not change the dimension of the tangent space.
This is not the case for the degree restriction or the basis-size restriction; the potential range of those values is unknown, and the existence of redundant basis polynomials can cause too early termination of the algorithm.

\section{Numerical experiments}\label{sec:experiments}
In this section, we evaluate the original, coefficient-normalized, and gradient-normalized VCAs. We perform two experiments. First, we compare the three algorithms regarding the size of bases and computational time. We ran the algorithms with generic points and a sufficiently small $\epsilon$ and then compared the size of output bases. Next, the stability of the algorithms against input perturbations and scaling is tested to demonstrate the advantage of scaling consistency. 

\paragraph{Implementation} All the experiments were run on a workstation with sixteen processors and 512GB memory. Our implementation of the algorithms is available in an open-source Python library MAVI\footnote{\url{https://github.com/HiroshiKERA/monomial-agnostic-vanishing-ideal}}, which supports three popular backends (Numpy~\citep{numpy}, JAX~\citep{jax}, and PyTorch~\citep{pytorch}) and is compatible with a large scale computation using GPUs. Sympy~\citep{sympy} was used for the calculation of coefficients. Here, we used the Numpy backend (CPU only). 

\subsection{Size of bases and runtime of the algorithms}\label{sec:generic-points}
We applied the three algorithms to sets of forty generic points of two-, three-, four-, and five-dimensional points with $\epsilon=10^{-6}$. Using generic points and small $\epsilon$ forces the algorithms to compute high-degree polynomials, which results in a large computational cost. Generic points are sampled uniformly from $[-1, 1]^n$, where $n = 2, 3, 4, 5$. 
We here present the result of a single run for each set of $n$-dimensional points because the results did not change among the several runs over random samplings of points in the preliminary experiment. Table~\ref{table:generic-points} summarizes the results. Three observations were obtained: (i) the size of the VCA bases was approximately 2-3 times larger\footnote{Note that in our VCA implementation, the degree-2 precandidate polynomials $C_2^{\text{pre}}$ is generated from the product set of $(F_1, F_1)$ without duplication (i.e., only $\binom{|F_1|}{2}$ polynomials are generated instead of $|F_1|^2$). Thus, computed VCA bases were more compact than the naive construction of $C_2^{\text{pre}}$.} than that of the others, (ii) both normalization (particularly coefficient normalization) took longer runtime as the $n$ increases, and (iii) the two normalization results in bases with the same configuration. Result~(i) indicates that the VCA suffers from the spurious vanishing problem. Results~(ii) and~(iii) show that gradient normalization runs much faster than coefficient normalization, and it is also a valid normalization method as coefficient normalization is. 

\begin{table}[t]\label{table:generic-points}
    \centering
    \caption{Summary of the bases given by three algorithms for sets of generic points. The labels \textit{+ coeff.} and \textit{+ grad.} denote the coefficient- and gradient-normalized VCA, respectively. Columns $|G|$, $|G_t|_t$, and $\textit{max deg.}$ denote the size of the whole and degree-$t$ bases and the highest degree of polynomials in the basis, respectively.
    The VCA resulted in a large size of bases, which implies the inclusion of the spuriously vanishing polynomials. Coefficient and gradient normalization resulted in bases of the same configuration. As for runtime, gradient normalization was by far more efficient than coefficient normalization.}
            \begin{tabular}{c||c|clcl}
        \toprule
            (\#points, dim.) & method & $|G|$ & $[|G_1|, |G_2|, \ldots]$ & max deg. & runtime [s] \\
        \midrule
        \multirow{3}{*}{(50, 2)}      & VCA             & 50  & [0, 0, 0, 2, 3, 4, 5, 6, 7, 13, 10]  & 10 &  5.81 $\times 10^{-1}$\\
                                      & \quad + coeff.  & \textbf{15}  & [0, 0, 0, 0, 0, 0, 0, 0, 0, 5, 10]   & 10 &  2.03\\
                                      & \quad + grad.   & \textbf{15}  & [0, 0, 0, 0, 0, 0, 0, 0, 0, 5, 10]   & 10 &  \textbf{1.87} $\mathbf{\times 10^{-2}}$\\
        \midrule
        \multirow{3}{*}{(50, 3)}      & VCA             & 98  & [0, 0, 0, 8, 15, 30, 45]             & 6 &  \textbf{3.00} $\mathbf{\times 10^{-3}}$\\
                                      & \quad + coeff.  & \textbf{40}  & [0, 0, 0, 0, 0, 6, 34]               & 6 &  8.99 \\
                                      & \quad + grad.   & \textbf{40}  & [0, 0, 0, 0, 0, 6, 34]               & 6 &  3.30 $\times 10^{-1}$\\
        \midrule
        \multirow{3}{*}{(50, 4)}      & VCA             & 145  & [0, 0, 0, 20, 65, 60]                 & 5 &  \textbf{1.85} $\mathbf{\times 10^{-2}}$\\
                                      & \quad + coeff.  & \textbf{80}  & [0, 0, 0, 0, 20, 60]  & 5 &  2.75 $\times 10$\\
                                      & \quad + grad.   & \textbf{80}  & [0, 0, 0, 0, 20, 60]  & 5 &  5.70 $\times 10^{-1}$ \\
        \midrule
        \multirow{3}{*}{(50, 5)}      & VCA             & 191  & [0, 0, 0, 46, 145]  & 4 &  \textbf{1.11} $\mathbf{\times 10^{-2}}$\\
                                      & \quad + coeff.    & \textbf{73}  & [0, 0, 0, 6, 76]  & 4 &  8.13 $\times 10$\\
                                      & \quad + grad.   & \textbf{73}  & [0, 0, 0, 6, 76]  & 4 &  2.94 \\
        \bottomrule
    \end{tabular}
\end{table}

\subsection{Robustness against input perturbations and scaling}\label{sec:experiment-scaling}
Here, we evaluate three algorithms in terms of robustness against input perturbations and scaling. 
We performed the following test. 
\begin{definition}[configuration retrieval test]
Let $G \subset\fullring$ be a finite set of polynomials, and let $T$ be the maximum degree of polynomials in $G$. An algorithm $\mathcal{A}$, which calculates a set of polynomials $\widehat{G}\subset \fullring$ from a set of points $X\subset\mathbb{R}^n$, is considered to successfully retrieve the configuration of $G$ if $\forall t = 0, 1, \ldots, T, |G_t| = |\widehat{G}_t|$.  
\end{definition}

The configuration retrieval test verifies if the algorithm outputs a set of polynomials with the same configuration as the target system up to the maximum degree of polynomials in the target system. 

\paragraph{Setup} We here consider this a necessary condition for a good approximate basis construction. 
We considered the following affine varieties $\{V_i\}_{i=1,2,3}$. 
\begin{align}
    V_1 &= \qty{(x,y) \in \mathbb{R}^2 \mid (x^2+y^2)^3 - 4x^2y^2 = 0}, \\
    V_2 &= \qty{(x, y, z) \in \mathbb{R}^3 \mid x + y - z = 0, x^3 - 9(x^2 - 3y^2) = 0}, \\
    V_3 &= \qty{(x, y, z) \in \mathbb{R}^3 \mid x^2-y^2z^2+z^3 = 0}.
\end{align}
We calculated Gr\"{o}bner basis and border basis for each variety and confirmed that these have the same configuration as each of those shown above, and used these configurations as the target ones.  
A hundred points (say, $X_i^{*}$) were sampled from each $V_i$.\footnote{The parametric representation is known for each $V_i$. A uniform sampling was performed in the parameter space. Refer to Appendix~\ref{sec:experiment-details}.} 
Each $X_i^{*}$ was preprocessed by subtracting the mean and scaling to make all the values range within [-1,1]. The sampled points were then perturbed by an additive Gaussian noise $\mathcal{N}(\mathbf{0}, \nu I)$, where $I$ denotes the identity matrix and $\nu\in \{0.05, 0.10\}$, and then recentered again. The set of such perturbed points from $X_i^{*}$ is denoted by $X_i$. Five scales $\alpha X_i, (\alpha=0.01, 0.1, 1.0, 10, 100)$ were considered.
Because of the perturbations, the choice of $\epsilon$ affects the computational results.
To circumvent the problem of choosing the problem of  $\epsilon$ of the algorithms, we performed a linear search. Thus, a run of an algorithm is considered to have passed the configuration retrieval test if there exists any such $\epsilon$. 
The linear search of $\epsilon$ was conducted with $[10^{-5}\alpha, \alpha)$ with a step size $10^{-3}\alpha$. We conducted twenty independent runs for each setting, changing the perturbation to $X_i^*$. 

\paragraph{Results} Tables~\ref{table:SRT-noise05} and~\ref{table:SRT-noise10} show the results with 5\% and 10\% noise, respectively. Gradient normalization exhibited superior robustness to the other two, and its success rates are all 1.0 for any scaling and dataset. In contrast, the VCA did not retrieve bases of the target configuration in any case because spurious approximately vanishing polynomials were included in the computed bases. Coefficient normalization succeeded for some scales (but not necessarily for $\alpha=1$), which implies that the preprocessing is crucial for this normalization. 
Another observation with gradient normalization is that the range of the valid $\epsilon$ and the extent of vanishing changes in proportion to $\alpha$. One might consider the extent of vanishing (the \textit{e.v.} columns) at $\alpha =100$ is large; however, this is only because the absolute level of the perturbations increased by the scaling. The signal-to-noise ratio of the extent of vanishing is consistent across all $\alpha$. In contrast, with coefficient normalization, the range of the valid $\epsilon$ and the extent of vanishing changes nonlinearly. In particular, for $\alpha = 10, 100$, these values are almost saturated. 

\begin{table*}
    \centering
    \caption{Summary of the configuration retrieval test of twenty independent runs with 5 \% noise. The labels \textit{+coeff.} and \textit{+grad.} denote the coefficient- and gradient-normalized VCA, respectively. Column \textit{e.v.} denotes the extent of vanishing at the unperturbed points. The values of range and the extent of vanishing are averaged values over twenty independent runs. As indicated by the success rate, the proposed gradient-weighted normalization approach is robust and consistent (see the proportional change in the range and the extent of vanishing) to the scaling, whereas coefficient normalization is not.}\label{table:SRT-noise05}
\begin{tabular}{c|crccc}
    \toprule
        dataset & method & scale $\alpha$ & range & e.v. & success rate \\
    \midrule
    \multirow{15}{*}{$V_1$} & \multirow{5}{*}{VCA}    & 0.01        & --                             & --                          & 0.00 [00/20]       \\
                       &                              & 0.1         & --                             & --                          & 0.00 [00/20]       \\
                       &                              & 1.0         & --                             & --                          & 0.00 [00/20]       \\
                       &                              & 10          & --                             & --                          & 0.00 [00/20]       \\
                       &                              & 100         & --                             & --                          & 0.00 [00/20]       \\
                       & \multirow{5}{*}{\quad + coeff. }    & 0.01        & --                             & --                          & 0.00 [00/20]        \\
                       &                              & 0.1         & --                             & --                          & 0.00 [00/20]        \\
                       &                              & 1.0         & [5.23, 5.73] $\times 10^{-3}$  & 4.56 $\times 10^{-3}$       & 0.90 [18/20]       \\
                       &                              & 10          & [4.99, 5.94] $\times 10^{+0}$  & 4.82 $\times 10^{+0}$       & \textbf{1.00 [20/20]}       \\
                       &                              & 100         & [5.35, 5.93] $\times 10^{+0}$  & 4.65 $\times 10^{+0}$       & 0.90 [18/20]       \\
                       & \multirow{5}{*}{\quad + grad. }    & 0.01         & [1.91, 2.26] $\times 10^{-3}$  & 2.05 $\times 10^{-3}$       & \textbf{1.00 [20/20]}       \\
                       &                              & 0.1         & [1.91, 2.26] $\times 10^{-2}$  & 2.05 $\times 10^{-2}$       & \textbf{1.00 [20/20]}       \\
                       &                              & 1.0         & [1.91, 2.26] $\times 10^{-1}$  & 2.05 $\times 10^{-1}$       & \textbf{1.00 [20/20]}       \\
                       &                              & 10          & [1.91, 2.26] $\times 10^{+0}$  & 2.05 $\times 10^{+0}$       & \textbf{1.00 [20/20]}       \\
                       &                              & 100         & [1.91, 2.26] $\times 10^{+1}$  & 2.05 $\times 10^{+1}$       & \textbf{1.00 [20/20]}       \\
    \midrule
    \multirow{15}{*}{$V_2$} & \multirow{5}{*}{VCA}    & 0.01        & --                             & --                          & 0.00 [00/20]       \\
                       &                              & 0.1         & --                             & --                          & 0.00 [00/20]       \\
                       &                              & 1.0         & --                             & --                          & 0.00 [00/20]       \\
                       &                              & 10          & --                             & --                          & 0.00 [00/20]       \\
                       &                              & 100         & --                             & --                          & 0.00 [00/20]       \\
                       & \multirow{5}{*}{\quad + coeff. }    & 0.01        & --                             & --                          & 0.00 [00/20]        \\
                       &                              & 0.1         & --                             & --                          & 0.00 [00/20]        \\
                       &                              & 1.0         & --                             & --                          & 0.00 [00/20]        \\
                       &                              & 10          & [1.96, 6.91] $\times 10^{+0}$  & 1.36 $\times 10^{+0}$       & \textbf{1.00 [20/20]}       \\
                       &                              & 100         & --                             & --                          & 0.00 [00/20]        \\
                       & \multirow{5}{*}{\quad + grad. }    & 0.01         & [1.43, 2.41] $\times 10^{-3}$  & 8.55 $\times 10^{-4}$       & \textbf{1.00 [20/20]}       \\
                       &                              & 0.1         & [1.43, 2.41] $\times 10^{-2}$  & 8.55 $\times 10^{-3}$       & \textbf{1.00 [20/20]}       \\
                       &                              & 1.0         & [1.43, 2.41] $\times 10^{-1}$  & 8.55 $\times 10^{-2}$       & \textbf{1.00 [20/20]}       \\
                       &                              & 10          & [1.43, 2.41] $\times 10^{+0}$  & 8.55 $\times 10^{-1}$       & \textbf{1.00 [20/20]}       \\
                       &                              & 100         & [1.43, 2.41] $\times 10^{+1}$  & 8.55 $\times 10^{+0}$       & \textbf{1.00 [20/20]}       \\
    \midrule
    \multirow{15}{*}{$V_3$} & \multirow{5}{*}{VCA}    & 0.01        & --                             & --                          & 0.00 [00/20]       \\
                       &                              & 0.1         & --                             & --                          & 0.00 [00/20]       \\
                       &                              & 1.0         & --                             & --                          & 0.00 [00/20]       \\
                       &                              & 10          & --                             & --                          & 0.00 [00/20]       \\
                       &                              & 100         & --                             & --                          & 0.00 [00/20]       \\
                       & \multirow{5}{*}{\quad + coeff. }    & 0.01        & --                             & --                          & 0.00 [00/20]        \\
                       &                              & 0.1         & [1.00, 1.00] $\times 10^{-6}$  & 1.58 $\times 10^{-7}$       & 0.05 [01/20]        \\
                       &                              & 1.0         & [1.23, 1.66] $\times 10^{-2}$  & 2.41 $\times 10^{-2}$       & \textbf{1.00 [20/20]}       \\
                       &                              & 10          & [3.94, 4.88] $\times 10^{+0}$  & 6.57 $\times 10^{+0}$       & 0.95 [19/20]        \\
                       &                              & 100         & [6.21, 6.77] $\times 10^{+0}$  & 7.68 $\times 10^{+0}$       & 0.85 [17/20]        \\
                       & \multirow{5}{*}{\quad + grad. }    & 0.01         & [1.43, 2.17] $\times 10^{-4}$  & 1.78 $\times 10^{-2}$       & \textbf{1.00 [20/20]}       \\
                       &                              & 0.1         & [1.11, 1.47] $\times 10^{-3}$  & 2.77 $\times 10^{-2}$       & \textbf{1.00 [20/20]}       \\
                       &                              & 1.0         & [1.11, 1.47] $\times 10^{-2}$  & 2.77 $\times 10^{-1}$       & \textbf{1.00 [20/20]}       \\
                       &                              & 10          & [1.11, 1.47] $\times 10^{-1}$  & 2.77 $\times 10^{+0}$       & \textbf{1.00 [20/20]}       \\
                       &                              & 100         & [1.11, 1.47] $\times 10^{+0}$  & 2.77 $\times 10^{+1}$       & \textbf{1.00 [20/20]}       \\

    \bottomrule
\end{tabular}
\end{table*}

\begin{table*}
    \centering
    \caption{Summary of the configuration retrieval test of twenty independent runs with 10\% noise. 
    }\label{table:SRT-noise10}
\begin{tabular}{c|crccc}
    \toprule
        dataset & method & scale $\alpha$ & range & e.v. & success rate \\
    \midrule
    \multirow{15}{*}{$V_1$} & \multirow{5}{*}{VCA}    & 0.01        & --                             & --                          & 0.00 [00/20]       \\
                       &                              & 0.1         & --                             & --                          & 0.00 [00/20]       \\
                       &                              & 1.0         & --                             & --                          & 0.00 [00/20]       \\
                       &                              & 10          & --                             & --                          & 0.00 [00/20]       \\
                       &                              & 100         & --                             & --                          & 0.00 [00/20]       \\
                       & \multirow{5}{*}{\quad + coeff. }    & 0.01        & --                             & --                          & 0.00 [00/20]        \\
                       &                              & 0.1         & --                             & --                          & 0.00 [00/20]        \\
                       &                              & 1.0         & [0.83, 1.01] $\times 10^{-2}$  & 7.62 $\times 10^{-3}$       & 0.95 [19/20]       \\
                       &                              & 10          & [5.21, 6.27] $\times 10^{+0}$  & 5.43 $\times 10^{+0}$       & \textbf{1.00 [20/20]}       \\
                       &                              & 100         & [5.50, 6.26] $\times 10^{+0}$  & 5.12 $\times 10^{+0}$       & 0.90 [18/20]       \\
                       & \multirow{5}{*}{\quad + grad. }    & 0.01         & [3.15, 3.64] $\times 10^{-3}$  & 3.59 $\times 10^{-3}$       & \textbf{1.00 [20/20]}       \\
                       &                              & 0.1         & [3.15, 3.64] $\times 10^{-2}$  & 3.59 $\times 10^{-2}$       & \textbf{1.00 [20/20]}       \\
                       &                              & 1.0         & [3.15, 3.64] $\times 10^{-1}$  & 3.59 $\times 10^{-1}$       & \textbf{1.00 [20/20]}       \\
                       &                              & 10          & [3.15, 3.64] $\times 10^{+0}$  & 3.59 $\times 10^{+0}$       & \textbf{1.00 [20/20]}       \\
                       &                              & 100         & [3.15, 3.64] $\times 10^{+1}$  & 3.59 $\times 10^{+1}$       & \textbf{1.00 [20/20]}       \\
    \midrule
    \multirow{15}{*}{$V_2$} & \multirow{5}{*}{VCA}    & 0.01        & --                             & --                          & 0.00 [00/20]       \\
                       &                              & 0.1         & --                             & --                          & 0.00 [00/20]       \\
                       &                              & 1.0         & --                             & --                          & 0.00 [00/20]       \\
                       &                              & 10          & --                             & --                          & 0.00 [00/20]       \\
                       &                              & 100         & --                             & --                          & 0.00 [00/20]       \\
                       & \multirow{5}{*}{\quad + coeff. }    & 0.01        & --                             & --                          & 0.00 [00/20]        \\
                       &                              & 0.1         & --                             & --                          & 0.00 [00/20]        \\
                       &                              & 1.0         & --                             & --                          & 0.00 [00/20]        \\
                       &                              & 10          & [4.95, 7.99] $\times 10^{+0}$  & 3.78 $\times 10^{+0}$       & \textbf{1.00 [20/20]}       \\
                       &                              & 100         & --                             & --                          & 0.00 [00/20]        \\
                       & \multirow{5}{*}{\quad + grad. }    & 0.01         & [3.47, 4.02] $\times 10^{-3}$  & 1.75 $\times 10^{-3}$       & \textbf{1.00 [20/20]}       \\
                       &                              & 0.1         & [3.47, 4.02] $\times 10^{-2}$  & 1.75 $\times 10^{-2}$       & \textbf{1.00 [20/20]}       \\
                       &                              & 1.0         & [3.47, 4.02] $\times 10^{-1}$  & 1.75 $\times 10^{-1}$       & \textbf{1.00 [20/20]}       \\
                       &                              & 10          & [3.47, 4.02] $\times 10^{+0}$  & 1.75 $\times 10^{+0}$       & \textbf{1.00 [20/20]}       \\
                       &                              & 100         & [3.47, 4.02] $\times 10^{+1}$  & 1.75 $\times 10^{+1}$       & \textbf{1.00 [20/20]}       \\
    \midrule
    \multirow{15}{*}{$V_3$} & \multirow{5}{*}{VCA}    & 0.01        & --                             & --                          & 0.00 [00/20]       \\
                       &                              & 0.1         & --                             & --                          & 0.00 [00/20]       \\
                       &                              & 1.0         & --                             & --                          & 0.00 [00/20]       \\
                       &                              & 10          & --                             & --                          & 0.00 [00/20]       \\
                       &                              & 100         & --                             & --                          & 0.00 [00/20]       \\
                       & \multirow{5}{*}{\quad + coeff. }    & 0.01        & --                             & --                          & 0.00 [00/20]        \\
                       &                              & 0.1         & [1.00, 1.00] $\times 10^{-6}$  & 1.18 $\times 10^{-6}$       & 0.20 [04/20]        \\
                       &                              & 1.0         & [1.27, 1.61] $\times 10^{-2}$  & 5.07 $\times 10^{-2}$       & \textbf{1.00 [20/20]}       \\
                       &                              & 10          & [3.89, 5.16] $\times 10^{+0}$  & 5.75 $\times 10^{+0}$       & \textbf{1.00 [20/20]}       \\
                       &                              & 100         & [4.61, 6.18] $\times 10^{+0}$  & 6.46 $\times 10^{+0}$       & \textbf{1.00 [20/20]}       \\
                       & \multirow{5}{*}{\quad + grad. }    & 0.01         & [0.95, 1.59] $\times 10^{-3}$  & 4.06 $\times 10^{-3}$       & \textbf{1.00 [20/20]}       \\
                       &                              & 0.1         & [0.95, 1.59] $\times 10^{-2}$  & 4.06 $\times 10^{-2}$       & \textbf{1.00 [20/20]}       \\
                       &                              & 1.0         & [0.95, 1.59] $\times 10^{-1}$  & 4.06 $\times 10^{-1}$       & \textbf{1.00 [20/20]}       \\
                       &                              & 10          & [0.95, 1.59] $\times 10^{+0}$  & 4.06 $\times 10^{+0}$       & \textbf{1.00 [20/20]}       \\
                       &                              & 100         & [0.95, 1.59] $\times 10^{+1}$  & 4.06 $\times 10^{+1}$       & \textbf{1.00 [20/20]}       \\

    \bottomrule
\end{tabular}
\end{table*}

\section{Conclusion}
This study proposed exploiting the gradient of polynomials and realized the first monomial-agnostic computation of generators of an ideal of points, aiming at spurious-vanishing-free approximate computation of vanishing ideals to bridge the gap between computational algebra and data-driven applications as machine learning. We introduced gradient normalization as an efficient and monomial-agnostic normalization and showed that its data-dependent nature provides novel properties such as scaling consistency, by which the basis of non-scaled points can be retrieved from that of translated or scaled points. We also showed that gradient normalization leads to a bounded increase in the extent of vanishing under perturbations, enabling us to select the threshold $\epsilon$ based on the geometrical intuition on the perturbations. Moreover, we demonstrated that although the gradients of polynomials at points are numerical entities, these can reveal some symbolic relations between polynomials. This result was used to reduce the redundancy of the basis. We believe that this work opens up new directions for research on data-driven monomial-agnostic computational algebra, where the existing notions, operations, and algorithms based on symbolic computation could be redefined or reformulated into a fully numerical computation.

\section*{Acknowledgement}
This work was supported by JST, ACT-X Grant Number JPMJAX200F, Japan.

\bibliographystyle{elsarticle-harv}

\begin{thebibliography}{30}
\expandafter\ifx\csname natexlab\endcsname\relax\def\natexlab#1{#1}\fi
\expandafter\ifx\csname url\endcsname\relax
  \def\url#1{\texttt{#1}}\fi
\expandafter\ifx\csname urlprefix\endcsname\relax\def\urlprefix{URL }\fi

\bibitem[{Abbott et~al.(2008)Abbott, Fassino, and Torrente}]{abbott2008stable}
Abbott, J., Fassino, C., Torrente, M.-L., 2008. Stable border bases for ideals
  of points. Journal of Symbolic Computation 43~(12), 883--894.

\bibitem[{Antonova et~al.(2020)Antonova, Maydanskiy, Kragic, Devlin, and
  Hofmann}]{antonova2020analytic}
Antonova, R., Maydanskiy, M., Kragic, D., Devlin, S., Hofmann, K., 2020.
  Analytic manifold learning: Unifying and evaluating representations for
  continuous control. arXiv preprint arXiv:2006.08718.

\bibitem[{Bradbury et~al.(2018)Bradbury, Frostig, Hawkins, Johnson, Leary,
  Maclaurin, Necula, Paszke, Vander{P}las, Wanderman-{M}ilne, and Zhang}]{jax}
Bradbury, J., Frostig, R., Hawkins, P., Johnson, M.~J., Leary, C., Maclaurin,
  D., Necula, G., Paszke, A., Vander{P}las, J., Wanderman-{M}ilne, S., Zhang,
  Q., 2018. {JAX}: composable transformations of {P}ython+{N}um{P}y programs.
\newline\urlprefix\url{http://github.com/google/jax}

\bibitem[{Fassino(2010)}]{fassino2010almost}
Fassino, C., 2010. Almost vanishing polynomials for sets of limited precision
  points. Journal of Symbolic Computation 45~(1), 19--37.

\bibitem[{Fassino and Torrente(2013)}]{fassino2013simple}
Fassino, C., Torrente, M.-L., 2013. Simple varieties for limited precision
  points. Theoretical Computer Science 479, 174--186.

\bibitem[{Harris et~al.(2020)Harris, Millman, van~der Walt, Gommers, Virtanen,
  Cournapeau, Wieser, Taylor, Berg, Smith, Kern, Picus, Hoyer, van Kerkwijk,
  Brett, Haldane, Fernández~del Río, Wiebe, Peterson, Gérard-Marchant,
  Sheppard, Reddy, Weckesser, Abbasi, Gohlke, and Oliphant}]{numpy}
Harris, C.~R., Millman, K.~J., van~der Walt, S.~J., Gommers, R., Virtanen, P.,
  Cournapeau, D., Wieser, E., Taylor, J., Berg, S., Smith, N.~J., Kern, R.,
  Picus, M., Hoyer, S., van Kerkwijk, M.~H., Brett, M., Haldane, A.,
  Fernández~del Río, J., Wiebe, M., Peterson, P., Gérard-Marchant, P.,
  Sheppard, K., Reddy, T., Weckesser, W., Abbasi, H., Gohlke, C., Oliphant,
  T.~E., 2020. Array programming with {NumPy}. Nature 585, 357–362.

\bibitem[{Heldt et~al.(2009)Heldt, Kreuzer, Pokutta, and
  Poulisse}]{heldt2009approximate}
Heldt, D., Kreuzer, M., Pokutta, S., Poulisse, H., 2009. Approximate
  computation of zero-dimensional polynomial ideals. Journal of Symbolic
  Computation 44~(11), 1566--1591.

\bibitem[{Hou et~al.(2016)Hou, Nie, and Tao}]{hou2016discriminative}
Hou, C., Nie, F., Tao, D., 2016. Discriminative vanishing component analysis.
  In: Proceedings of the Thirtieth AAAI Conference on Artificial Intelligence
  (AAAI). AAAI Press, pp. 1666--1672.

\bibitem[{Iraji and Chitsaz(2017)}]{iraji2017principal}
Iraji, R., Chitsaz, H., 2017. Principal variety analysis. In: Proceedings of
  the 1st Annual Conference on Robot Learning (ACRL). PMLR, pp. 97--108.

\bibitem[{Karimov et~al.(2020)Karimov, Nepomuceno, Tutueva, and
  Butusov}]{karimov2020algebraic}
Karimov, A., Nepomuceno, E.~G., Tutueva, A., Butusov, D., 2020. Algebraic
  method for the reconstruction of partially observed nonlinear systems using
  differential and integral embedding. Mathematics 8~(2).

\bibitem[{Karimov et~al.(2023)Karimov, Rybin, Kopets, Karimov, Nepomuceno, and
  Butusov}]{karimov2023identifying}
Karimov, A., Rybin, V., Kopets, E., Karimov, T., Nepomuceno, E., Butusov, D.,
  2023. Identifying empirical equations of chaotic circuit from data. Nonlinear
  Dynamics 111~(1), 871--886.

\bibitem[{Kera(2022)}]{kera2022border}
Kera, H., 2022. Border basis computation with gradient-weighted normalization.
  In: Proceedings of the 2022 International Symposium on Symbolic and Algebraic
  Computation (ISSAC). ISSAC '22. Association for Computing Machinery, pp.
  225--234.

\bibitem[{Kera and Hasegawa(2016)}]{kera2016noise}
Kera, H., Hasegawa, Y., 2016. Noise-tolerant algebraic method for
  reconstruction of nonlinear dynamical systems. Nonlinear Dynamics 85~(1),
  675--692.

\bibitem[{Kera and Hasegawa(2018)}]{kera2018approximate}
Kera, H., Hasegawa, Y., 2018. Approximate vanishing ideal via data knotting.
  In: Proceedings of the Thirty-Second {AAAI} Conference on Artificial
  Intelligence (AAAI). AAAI Press, pp. 3399--3406.

\bibitem[{{Kera} and {Hasegawa}(2019)}]{kera2019spurious}
{Kera}, H., {Hasegawa}, Y., 2019. Spurious vanishing problem in approximate
  vanishing ideal. IEEE Access 7, 178961--178976.

\bibitem[{Kera and Hasegawa(2020)}]{kera2020gradient}
Kera, H., Hasegawa, Y., 2020. Gradient boosts the approximate vanishing ideal.
  In: Proceedings of the Thirty-Fourth {AAAI} Conference on Artificial
  Intelligence (AAAI). AAAI Press, pp. 4428--4425.

\bibitem[{Kera and Iba(2016)}]{kera2016vanishing}
Kera, H., Iba, H., 2016. Vanishing ideal genetic programming. In: Proceedings
  of the 2016 IEEE Congress on Evolutionary Computation (CEC). IEEE, pp.
  5018--5025.

\bibitem[{Kir{\'a}ly et~al.(2014)Kir{\'a}ly, Kreuzer, and
  Theran}]{kiraly2014dual}
Kir{\'a}ly, F.~J., Kreuzer, M., Theran, L., 2014. Dual-to-kernel learning with
  ideals. arXiv preprint arXiv:1402.0099.

\bibitem[{Limbeck(2013)}]{limbeck2013computation}
Limbeck, J., 2013. Computation of approximate border bases and applications.
  Ph.D. thesis, Passau, Universit{\"a}t Passau.

\bibitem[{Livni et~al.(2013)Livni, Lehavi, Schein, Nachliely, Shalev-Shwartz,
  and Globerson}]{livni2013vanishing}
Livni, R., Lehavi, D., Schein, S., Nachliely, H., Shalev-Shwartz, S.,
  Globerson, A., 2013. Vanishing component analysis. In: Proceedings of the
  Thirteenth International Conference on Machine Learning (ICML). PMLR, pp.
  597--605.

\bibitem[{Meurer et~al.(2017)Meurer, Smith, Paprocki, \v{C}ert\'{i}k,
  Kirpichev, Rocklin, Kumar, Ivanov, Moore, Singh, Rathnayake, Vig, Granger,
  Muller, Bonazzi, Gupta, Vats, Johansson, Pedregosa, Curry, Terrel,
  Rou\v{c}ka, Saboo, Fernando, Kulal, Cimrman, and Scopatz}]{sympy}
Meurer, A., Smith, C.~P., Paprocki, M., \v{C}ert\'{i}k, O., Kirpichev, S.~B.,
  Rocklin, M., Kumar, A., Ivanov, S., Moore, J.~K., Singh, S., Rathnayake, T.,
  Vig, S., Granger, B.~E., Muller, R.~P., Bonazzi, F., Gupta, H., Vats, S.,
  Johansson, F., Pedregosa, F., Curry, M.~J., Terrel, A.~R., Rou\v{c}ka, v.,
  Saboo, A., Fernando, I., Kulal, S., Cimrman, R., Scopatz, A., 2017. {SymPy}:
  symbolic computing in {Python}. PeerJ Computer Science 3, e103.

\bibitem[{Paszke et~al.(2019)Paszke, Gross, Massa, Lerer, Bradbury, Chanan,
  Killeen, Lin, Gimelshein, Antiga, Desmaison, Kopf, Yang, DeVito, Raison,
  Tejani, Chilamkurthy, Steiner, Fang, Bai, and Chintala}]{pytorch}
Paszke, A., Gross, S., Massa, F., Lerer, A., Bradbury, J., Chanan, G., Killeen,
  T., Lin, Z., Gimelshein, N., Antiga, L., Desmaison, A., Kopf, A., Yang, E.,
  DeVito, Z., Raison, M., Tejani, A., Chilamkurthy, S., Steiner, B., Fang, L.,
  Bai, J., Chintala, S., 2019. {PyTorch}: An imperative style, high-performance
  deep learning library. In: Proceedings of the Thirty-Second Advances in
  Neural Information Processing Systems (NeurIPS). Curran Associates, Inc., pp.
  8024--8035.

\bibitem[{Robbiano and Abbott(2010)}]{robbiano2010approximate}
Robbiano, L., Abbott, J., 2010. Approximate {Commutative} {Algebra}. Vol.~1.
  Springer.

\bibitem[{Torrente(2008)}]{torrente2009application}
Torrente, M.-L., 2008. Application of algebra in the oil industry. Ph.D.
  thesis, Scuola Normale Superiore, Pisa.

\bibitem[{Wang and Ohtsuki(2018)}]{wang2018nonlinear}
Wang, L., Ohtsuki, T., 2018. Nonlinear blind source separation unifying
  vanishing component analysis and temporal structure. IEEE Access 6,
  42837--42850.

\bibitem[{Wang et~al.(2019)Wang, Li, Li, and Xu}]{wang2019polynomial}
Wang, Z., Li, Q., Li, G., Xu, G., 2019. Polynomial representation for
  persistence diagram. In: Proceedings of the 2019 IEEE/CVF Conference on
  Computer Vision and Pattern Recognition (CVPR). pp. 6116--6125.

\bibitem[{Wirth et~al.(2023)Wirth, Kera, and Pokutta}]{wirth2023approximate}
Wirth, E.~S., Kera, H., Pokutta, S., 2023. Approximate vanishing ideal
  computations at scale. In: The Eleventh International Conference on Learning
  Representations (ICLR).
\newline\urlprefix\url{https://openreview.net/forum?id=3ZPESALKXO}

\bibitem[{Wirth and Pokutta(2022)}]{wirth2022conditional}
Wirth, E.~S., Pokutta, S., 28--30 Mar 2022. Conditional gradients for the
  approximately vanishing ideal. In: Proceedings of The 25th International
  Conference on Artificial Intelligence and Statistics (AISTATS). Vol. 151 of
  Proceedings of Machine Learning Research. PMLR, pp. 2191--2209.

\bibitem[{Yan et~al.(2018)Yan, Yan, Xiao, Wang, and Zuo}]{yan2018deep}
Yan, H., Yan, Z., Xiao, G., Wang, W., Zuo, W., 2018. Deep vanishing component
  analysis network for pattern classification. Neurocomputing 316, 240--250.

\bibitem[{Zhao and Song(2014)}]{zhao2014hand}
Zhao, Y.-G., Song, Z., 2014. Hand posture recognition using approximate
  vanishing ideal generators. In: Proceedings of the 2014 IEEE International
  Conference on Image Processing (ICIP). IEEE, pp. 1525--1529.

\end{thebibliography}

\appendix

\renewcommand{\thesection}{\Alph{section}}

\section{A preprocessing method for faster gradient normalization}\label{app:preprocessing}
\begin{algorithm} (Preprocessing for accelerated gradient normalization)
Let $X\subset\mathbb{R}^n$ be a set of points and $\epsilon \ge 0$ be a fixed threshold. Perform the following procedures. 
\begin{enumerate}
    \item[\textbf{P1}] From $X$, calculate the mean-subtracted points $X_0$. 
    \item[\textbf{P2}] Perform the singular value decomposition (SVD) on $X_0$ to obtain $X_0 = UDV^{\top}$. 
    Let $U = \mqty(U_{\mathrm{F}} & U_{\mathrm{G}})$ and $U = \mqty(V_{\mathrm{F}} & V_{\mathrm{G}})$, where ${\,\cdot\,}_{\mathrm{F}}$ and ${\,\cdot\,}_{\mathrm{G}}$ denote the singular vectors grater than $\epsilon$ and others, respectively. 
    \item[\textbf{P3}] Return $U_{\mathrm{F}}$, $U_{\mathrm{G}}$, $V_{\mathrm{F}}$, and $V_{\mathrm{G}}$. 
\end{enumerate}
\end{algorithm}
Note that if we additionally perform $F_1 = C_1V_{\mathrm{F}}$ and $G_1 = C_1V_{\mathrm{G}}$, this procedure is equivalent to the normalized VCA (with coefficient or gradient normalization) at degree $t=1$. Thus, this preprocessing is done for free.
Now, we use new indeterminates $\{y_1, y_2, \ldots, y_m\} := \{x_1, x_2, \ldots, x_n\}V_{\mathrm{F}}$, a polynomial ring $\mathbb{R}[y_1, y_2, \ldots, y_\nu]$, and a set of points $Y = XV_{\mathrm{F}}$. The gradient-normalized VCA is then performed on $Y$. Let $(F[y_1, \ldots, y_\nu], G[y_1, \ldots, y_\nu])$ be the outputs of the algorithm. 
At inferance time, given a new set of points $X_{\mathrm{new}}$, we first transform it as $Y_{\mathrm{new}}= X_{\mathrm{new}}V_{\mathrm{F}}$ and then feed it to the computed polynomials $F[y_1, \ldots, y_\nu]$, $G[y_1, \ldots, y_\nu]$. At the same time, $X_{\mathrm{new}}$ is fed to a set of linar approximately vanishing polynomials $G_1[x_1, x_2, \ldots, x_n] = \{x_1, x_2,\ldots, x_n\}V_{\mathrm{G}}$.
As the bottleneck of the gradient normalization is the calculation of $\gnvec{X}{C_t}^{\top}\gnvec{X}{C_t}$, which costs $O(n^2\nu |C_t|)$, this part of the algorithm can get $(n/\nu )^2$-times acceleration by the preprocessing.

\section{Experiment details}\label{sec:experiment-details}
We describe the details of the experiment setup of Section~\ref{sec:generic-points}. The parametric representation of the three varieties is as follows. 
\begin{align}
    V_1 &= \qty{(x,y) \in \mathbb{R}^2 \mid x = \cos(2u)\cos(u), y = \cos(2u)\sin(u), u\in \mathbb{R}}, \\
    V_2 &= \qty{(x, y, z) \in \mathbb{R}^3 \mid x = 3(3-u^2), y = u(3-u^2), z = x + y, u\in \mathbb{R}}, \\
    V_3 &= \qty{(x, y, z) \in \mathbb{R}^3 \mid x = v(u^2-v^2), y=u, z=u^2-v^2, (u, v)\in \mathbb{R}^2}.
\end{align}
In the experiments, the points were sampled uniformly in the parameter space. In particular, we sampled a hundred points from $u \in [-1, 1)$ for $V_1$, $u \in [-2.5, 2.5)$ for $V_2$, and $(u, t) \in [-1, 1)^2$ for $V_3$. To explore the valid range of $\epsilon$ for approximate configuration retrieval, linear searching is conducted. Basis construction is repeatedly performed with $\epsilon \in [10^{-5}\alpha, \alpha)$ with a step size $10^{-3}\alpha$, where $\alpha$ denotes the scaling factor in the experiment (i.e., $\alpha \in \{0.1, 1.0, 10\}$). We set the constant polynomial $f_0 = 1/\sqrt{m}, 1, \sum_{\bm{x}\in X}\norm{\bm{x}}_{\infty}/m$ for original, coefficient-normalized, and gradient-normalized VCA for the respective consistency to the normalization.
For gradient normalization, we also used $\gamma=1/\sqrt{m}$.

\end{document}